\documentclass{ecai} 

\usepackage{latexsym}
\usepackage{amssymb}
\usepackage{amsmath}
\usepackage{amsthm}
\usepackage{booktabs}
\usepackage{enumitem}
\usepackage{graphicx}
\usepackage{color}
\usepackage{stmaryrd} 

\usepackage{subcaption}  

\usepackage{algorithm}
\usepackage{algorithmic}
\usepackage{url}
\usepackage[hidelinks]{hyperref} 

\newtheorem{theorem}{Theorem}
\newtheorem{lemma}[theorem]{Lemma}

\newtheorem{definition}[theorem]{Definition}

\DeclareMathOperator*{\argmin}{arg\,min}

\newcommand{\BibTeX}{B\kern-.05em{\sc i\kern-.025em b}\kern-.08em\TeX}

\begin{document}


\begin{frontmatter}


\paperid{2455} 


\title{Last-Iterate Convergence in Adaptive Regret Minimization for Approximate Extensive-Form Perfect Equilibrium}


\author[A]{\fnms{Hang}~\snm{Ren}}
\author[A]{\fnms{Xiaozhen}~\snm{Sun}}
\author[A]{\fnms{Tianzi}~\snm{Ma}}
\author[A,B]{\fnms{Jiajia}~\snm{Zhang}}
\author[A,B]{\fnms{Xuan}~\snm{Wang}\thanks{Corresponding Author. Email: wangxuan@cs.hitsz.edu.cn.}}

\address[A]{School of Computer Science and Technology, Harbin Institute of Technology (Shenzhen), China}
\address[B]{Guangdong Provincial Key Laboratory of Novel Security Intelligence Technologies}


\begin{abstract}
The Nash Equilibrium (NE) assumes rational play in imperfect-information Extensive-Form Games (EFGs) but fails to ensure optimal strategies for off-equilibrium branches of the game tree, potentially leading to suboptimal outcomes in practical settings. To address this, the Extensive-Form Perfect Equilibrium (EFPE), a refinement of NE, introduces controlled perturbations to model potential player errors. However, existing EFPE-finding algorithms, which typically rely on average strategy convergence and fixed perturbations, face significant limitations: computing average strategies incurs high computational costs and approximation errors, while fixed perturbations create a trade-off between NE approximation accuracy and the convergence rate of NE refinements.

To tackle these challenges, we propose an efficient adaptive regret minimization algorithm for computing approximate EFPE, achieving last-iterate convergence in two-player zero-sum EFGs. Our approach introduces Reward Transformation Counterfactual Regret Minimization (RTCFR) to solve perturbed games and defines a novel metric, the Information Set Nash Equilibrium (ISNE), to dynamically adjust perturbations. Theoretical analysis confirms convergence to EFPE, and experimental results demonstrate that our method significantly outperforms state-of-the-art algorithms in both NE and EFPE-finding tasks.
\end{abstract}
\end{frontmatter}

\section{Introduction}

The Nash Equilibrium (NE) is a foundational solution concept for imperfect-information games, assuming perfectly rational players who assign zero probability to suboptimal actions, or ``mistakes''. However, this assumption leads to arbitrary strategies at unreached decision nodes, undermining NE's robustness in real-world scenarios where players may deviate from optimal play~\cite{selten1975reexamination}. Such deviations can significantly degrade NE strategy performance against irrational opponents.

To address this limitation, \citet{selten1975reexamination} introduced the concept of \emph{trembles} or \emph{perturbations}, refining Nash Equilibrium (NE) by requiring each action to be selected with a minimum probability \(\epsilon\). This leads to refined equilibria, such as the \emph{Extensive-Form Perfect Equilibrium (EFPE)}~\cite{selten1975reexamination} and the Quasi-Perfect Equilibrium (QPE)~\cite{van1984relation}. EFPE accounts for potential mistakes by all players, whereas QPE considers only opponents' future mistakes, excluding the player's own future errors. As \(\epsilon \to 0\), these refinements converge to exact equilibria.

Algorithms for computing NE refinements typically solve perturbed games by extending NE-finding methods with constraints enforcing minimum perturbation probabilities. Early approaches, such as the linear programming (LP) method proposed by \citet{miltersen2010computing}, incorporated additional constraints. More recent iterative algorithms~\cite{farina2017regret, kroer2017smoothing, farina2018practical} rely on average strategy convergence but often face high computational costs, increased memory requirements, or approximation errors~\cite{brown2019deep}.

Recent advances in last-iterate convergence techniques, including optimistic methods~\cite{lee2021last, wei2020linear} and regularization approaches like \emph{Reward Transformation (RT)}~\cite{perolat2021poincare, abe2022mutation, abe2022last, meng2023efficient}, offer promising directions for improvement. However, their application to NE refinement remains underexplored, both theoretically and practically.

A key limitation of existing methods is their reliance on fixed perturbation schemes throughout training, which results in a suboptimal trade-off between NE approximation accuracy and the convergence rate of NE refinements~\cite{farina2017regret, kroer2017smoothing}. Developing an efficient adaptive perturbation framework that dynamically adjusts perturbations based on the state of NE convergence remains an open challenge.

In this paper, we propose an efficient last-iterate convergence algorithm for computing EFPE in two-player zero-sum Extensive-Form Games (EFGs) with imperfect information. Building on behavioral constraints in perturbed games~\cite{farina2017regret} and RT-based regret minimization~\cite{meng2023efficient}, we introduce \emph{Reward Transformation Counterfactual Regret Minimization (RTCFR)}, a novel method for solving \(\epsilon\)-perturbed EFGs using last-iterate strategies. Our approach maps the original and perturbed strategy spaces via distinct basis matrices determined by the perturbation settings, and we prove its asymptotic convergence to \(\epsilon\)-EFPE.

To enhance efficiency, we propose the \emph{Information Set Nash Equilibrium (ISNE)} metric, which dynamically adjusts perturbations during iterations to balance NE approximation accuracy and EFPE convergence rate. Experimental results demonstrate that our approach significantly outperforms state-of-the-art algorithms in both NE and EFPE computation tasks.


\section{Related Work}\label{sec:related_work}

\subsection{Last-Iterate Convergence Algorithms}

Last-iterate convergence in saddle-point problems is primarily achieved through two approaches: optimistic methods and regularization techniques.

Optimistic methods, such as Optimistic Gradient Descent Ascent (OGDA) and Optimistic Multiplicative Weights Update (OMWU), have demonstrated last-iterate convergence in Normal-Form Games (NFGs)~\cite{wei2020linear} and Extensive-Form Games (EFGs)~\cite{lee2021last, anagnostides2022last, liu2022power}. However, these methods often require small hyperparameters to ensure theoretical guarantees, leading to slow practical convergence. Empirically increasing these parameters can degrade performance, as observed in games like Leduc Poker~\cite{lee2021last}.

Regularization-based methods incorporate regularization terms into the reward function via \emph{reward transformation (RT)}, which mitigates divergence and breaks Poincaré recurrence cycles in iterative strategy updates~\cite{mertikopoulos2018cycles, perolat2021poincare}. For instance, \citet{perolat2022mastering} applied RT in Deep Nash to tackle large-scale games like Stratego using deep learning. Building on this, \citet{abe2022last, abe2022mutation} integrated RT into Follow-The-Regularized-Leader (FTRL) and Multiplicative Weights Update (MWU) for NFGs, guiding strategies toward Nash equilibrium by quantifying deviations between reference and current strategies. \citet{meng2023efficient} extended RT to regret minimization, proposing Reward Transformation Regret Matching+ (RTRM+) and RTCFR+ for discrete game problems. However, they did not provide theoretical convergence guarantees for RTCFR+ in EFGs due to technical challenges.

\subsection{Nash Equilibrium Refinement Algorithms}

Computing Nash equilibrium (NE) refinements typically involves perturbing the game to enforce minimum action probabilities, followed by applying NE-finding algorithms. Linear Programming (LP)-based methods, such as those by \citet{miltersen2010computing}, compute Quasi-Perfect Equilibrium (QPE) by perturbing the realization plan. Subsequent work~\cite{farina2017extensive, farina2018practical} extended LP-based approaches to Extensive-Form Perfect Equilibrium (EFPE) and QPE, but these methods are limited to small games due to the exponential growth of strategy representations and constraints.

Iterative algorithms, such as the Excessive Gap Technique (EGT), Online Mirror Descent (OMD), and Counterfactual Regret Minimization (CFR), offer scalable alternatives for approximate NE refinements. \citet{farina2017regret} developed a CFR-based method that maps behavioral strategies to a perturbed space using affine transformations. \citet{kroer2017smoothing} proposed an EGT-based approach building on the realization plan framework of \citet{miltersen2010computing}. More recently, \citet{bernasconi2024learning} introduced an optimistic OMD-based method for EFPE with last-iterate convergence.


\section{Preliminaries}\label{sec:preliminaries}

\subsection{Mathematical Notational Conventions}
\begin{itemize}
    \item Vectors and matrices are denoted in bold (e.g., \(\boldsymbol{x}\), \(\boldsymbol{M}\)).
    \item \(\mathbb{R}\) represents the set of real numbers, and \(\mathbb{N} = \{1, 2, \dots\}\) denotes the positive integers.
    \item For a finite set \(S = \{s_1, \dots, s_n\}\), \(\mathbb{R}^S\) (resp., \(\mathbb{R}^S_{\geq 0}\)) denotes the set of real (resp., non-negative real) \(|S|\)-dimensional vectors with entries \(\boldsymbol{x}[s_1], \dots, \boldsymbol{x}[s_n]\). The simplex over \(S\) is \(\Delta^S := \{\boldsymbol{x} \in \mathbb{R}^S_{\geq 0} : \sum_{s \in S} \boldsymbol{x}[s] = 1\}\).
    \item For finite sets \(S, S'\), \(\mathbb{R}^{S \times S'}\) (resp., \(\mathbb{R}^{S \times S'}_{\geq 0}\)) denotes the set of real (resp., non-negative real) \(|S| \times |S'|\) matrices \(\boldsymbol{M}\) with entries \(\boldsymbol{M}[s_r, s_c]\) for \(s_r \in S\), \(s_c \in S'\), where \(s_r\) and \(s_c\) index rows and columns, respectively.
\end{itemize}

\subsection{Normal-Form and Extensive-Form Games}

\begin{definition}[Normal-Form Game]
    A two-player zero-sum \emph{Normal-Form Game (NFG)} is a tuple \(\Gamma = (A_1, A_2, u)\), where \(A_1\) and \(A_2\) are the action sets of players 1 and 2, respectively, and \(u: A_1 \times A_2 \to \mathbb{R}\) is the utility function for player 1. The utility for player 2 is \(-u\).
\end{definition}

Given action sets \(A_1 = \{a_{1,1}, \dots, a_{1,n}\}\) and \(A_2 = \{a_{2,1}, \dots, a_{2,m}\}\), the utility function \(u\) is represented by a matrix \(\boldsymbol{U} \in \mathbb{R}^{A_1 \times A_2}\) with entries \(\boldsymbol{U}[a_{1,i}, a_{2,j}] = u(a_{1,i}, a_{2,j})\). A strategy for player \(i \in \{1, 2\}\) is a probability distribution over actions, denoted \(\boldsymbol{x}_i \in \Delta^{A_i}\).

\begin{definition}[Bilinear Game]
    A two-player zero-sum \emph{Bilinear Game} is a tuple \(\Gamma = (\mathcal{X}_1, \mathcal{X}_2, \boldsymbol{U})\), where \(\mathcal{X}_1 \subseteq \mathbb{R}^n\) and \(\mathcal{X}_2 \subseteq \mathbb{R}^m\) are convex, compact strategy spaces for players 1 and 2, respectively, and \(\boldsymbol{U} \in \mathbb{R}^{n \times m}\) is the utility matrix. The expected utility for player 1 is a biaffine function \(u: \mathcal{X}_1 \times \mathcal{X}_2 \to \mathbb{R}\), given by \(u(\boldsymbol{x}_1, \boldsymbol{x}_2) = \boldsymbol{x}_1^\top \boldsymbol{U} \boldsymbol{x}_2\). The utility for player 2 is \(-u(\boldsymbol{x}_1, \boldsymbol{x}_2)\).
\end{definition}

An \emph{Extensive-Form Game (EFG)} models sequential decision-making via a game tree with decisions at information sets, defined as follows:

\begin{definition}[Extensive-Form Game]
    A two-player zero-sum \emph{Extensive-Form Game (EFG)} is a tuple \(\Gamma = (H, Z, A, P, \mathcal{I}_i, \boldsymbol{x}_i, u_i)\) for \(i \in \{1, 2\}\), where:
    \begin{itemize}
        \item \(H\): The set of states, including the initial state \(\emptyset\).
        \item \(Z \subset H\): The set of terminal states (leaf nodes).
        \item \(A(h)\): The set of actions available at non-terminal state \(h \in H \setminus Z\).
        \item \(P: H \setminus Z \to \{0, 1, 2\}\): A function assigning each non-terminal state to a player, where \(P(h) = 0\) denotes the chance player.
        \item \(\mathcal{I}_i\): The information partition for player \(i\), grouping indistinguishable non-terminal states \(h \in H\) with \(P(h) = i\). For any \(h, h' \in I \in \mathcal{I}_i\), \(A(h) = A(h')\), and we denote \(A(I) := A(h)\) for any \(h \in I\).
        \item \(\boldsymbol{x}_i \in \mathbb{R}_{\geq 0}^{\cup_{I \in \mathcal{I}_i} A(I)}\): The behavioral strategy for player \(i\), where \(\boldsymbol{x}_i(I) \in \Delta^{A(I)}\) is a probability distribution over actions at information set \(I \in \mathcal{I}_i\), and \(\boldsymbol{x}_i[Ia]\) denotes the probability of action \(a \in A(I)\).
        \item \(u_i: Z \to \mathbb{R}\): The utility function for player \(i\) at terminal state \(z \in Z\).
    \end{itemize}
\end{definition}
We assume \emph{perfect recall}, ensuring players retain all prior information. Hereafter, \emph{EFG} refers to a two-player zero-sum extensive-form game with imperfect information and perfect recall.

\subsection{Sequences and Sequence-Form Strategy Representation}

A \emph{sequence} \(\sigma = (I, a)\), denoted \(Ia\), represents action \(a \in A(I)\) at information set \(I \in \mathcal{I}_i\). The set of sequences for player \(i\) is \(\Sigma_i := \{\emptyset\} \cup \{(I, a) : I \in \mathcal{I}_i, a \in A(I)\}\), where \(\emptyset\) is the empty sequence. For an information set \(I \in \mathcal{I}_i\), the \emph{parent sequence} \(pI \in \Sigma_i\) is the last sequence on the path from the root to \(I\), or \(\emptyset\) if player \(i\) has not acted before \(I\). A partial order \(I'a' \prec Ia\) indicates that the path from the root to action \(a\) at \(I\) passes through action \(a'\) at \(I'\); \(I'a' \preceq Ia\) means \(I'a' \prec Ia\) or \(I'a' = Ia\).

A \emph{sequence-form strategy}~\cite{romanovskii1962reduction, koller1996efficient, von1996efficient} is a realization plan, represented by a vector \(\boldsymbol{q}_i \in \mathcal{Q}_i \subseteq \mathbb{R}_{\geq 0}^{\Sigma_i}\). Given a behavioral strategy \(\boldsymbol{x}_i \in \mathcal{X}_i \subseteq \mathbb{R}_{\geq 0}^{\Sigma_i}, \boldsymbol{x}_i[\emptyset]=1\), the sequence-form strategy is computed as:
\begin{equation}
    \boldsymbol{q}_i[Ia] = \prod_{I'a'\in \Sigma_i: I'a' \preceq Ia} \boldsymbol{x}_i[I'a'],
    \label{eq:q}
\end{equation}
with \(\boldsymbol{q}_i[\emptyset] = 1\). Conversely:
\begin{equation}
    \boldsymbol{x}_i[Ia] = \frac{\boldsymbol{q}_i[Ia]}{\boldsymbol{q}_i[pI]},
    \label{eq:x}
\end{equation}
where \(\boldsymbol{q}_i[pI] = \sum_{a \in A(I)} \boldsymbol{q}_i[Ia]\). Equations~\eqref{eq:q} and \eqref{eq:x} enable conversion between sequence-form and behavioral strategies.

This representation allows an EFG to be formulated as a bilinear game \(\Gamma = (\mathcal{Q}_1, \mathcal{Q}_2, \boldsymbol{U})\), expressed as a \emph{bilinear saddle-point problem (BSPP)} for the Nash equilibrium~\cite{koller1996efficient}:
\begin{equation}
    \max_{\boldsymbol{q}_1 \in \mathcal{Q}_1} \min_{\boldsymbol{q}_2 \in \mathcal{Q}_2} \boldsymbol{q}_1^\top \boldsymbol{U} \boldsymbol{q}_2 = \min_{\boldsymbol{q}_2 \in \mathcal{Q}_2} \max_{\boldsymbol{q}_1 \in \mathcal{Q}_1} \boldsymbol{q}_1^\top \boldsymbol{U} \boldsymbol{q}_2,
    \label{eq:bspp}
\end{equation}
where \(\boldsymbol{U} \in \mathbb{R}^{\Sigma_1 \times \Sigma_2}\) is a sparse utility matrix with non-zero entries \(\boldsymbol{U}[\sigma_1, \sigma_2] = q_0(z)u(z)\) for sequence pairs \((\sigma_1, \sigma_2) \in \Sigma_1 \times \Sigma_2\) reaching a leaf node \(z \in Z\), and \(q_0(z)\) is the chance probability along the path.

\subsection{Nash Equilibrium and Refinements}

\begin{definition}[Nash Equilibrium]
    In an EFG \(\Gamma = (\mathcal{Q}_1, \mathcal{Q}_2, \boldsymbol{U})\), a strategy profile \(\boldsymbol{q}^* = (\boldsymbol{q}_1^*, \boldsymbol{q}_2^*) \in \mathcal{Q}_1 \times \mathcal{Q}_2\) is a Nash equilibrium (NE) if:
    \[
        \boldsymbol{q}_1^* = \arg\max_{\boldsymbol{q}_1 \in \mathcal{Q}_1} \min_{\boldsymbol{q}_2 \in \mathcal{Q}_2} \boldsymbol{q}_1^\top \boldsymbol{U} \boldsymbol{q}_2, \quad \boldsymbol{q}_2^* = \arg\min_{\boldsymbol{q}_2 \in \mathcal{Q}_2} \max_{\boldsymbol{q}_1 \in \mathcal{Q}_1} \boldsymbol{q}_1^\top \boldsymbol{U} \boldsymbol{q}_2.
    \]
\end{definition}

An \emph{approximate Nash equilibrium} is a strategy profile close to an exact NE, quantified by \emph{exploitability}:

\begin{definition}[Exploitability]
    In an EFG \(\Gamma = (\mathcal{Q}_1, \mathcal{Q}_2, \boldsymbol{U})\), the exploitability of a strategy profile \(\boldsymbol{q} = (\boldsymbol{q}_1, \boldsymbol{q}_2)\) is:
    \begin{equation}
        \text{Exp}(\boldsymbol{q}) = \max_{\boldsymbol{q}_1' \in \mathcal{Q}_1} \boldsymbol{q}_1'^\top \boldsymbol{U} \boldsymbol{q}_2 - \min_{\boldsymbol{q}_2' \in \mathcal{Q}_2} \boldsymbol{q}_1^\top \boldsymbol{U} \boldsymbol{q}_2'.
        \label{eq:exp}
    \end{equation}
\end{definition}

The \emph{Extensive-Form Perfect Equilibrium (EFPE)} is defined in a perturbed game where each action has a minimum probability \(\epsilon\). For any information set \(I \in \mathcal{I}_i\), the behavioral strategy satisfies:
\begin{equation}
    \sum_{a \in A(I)} \boldsymbol{x}_i[Ia] = 1, \quad \boldsymbol{x}_i[Ia] \geq \epsilon.
    \label{eq:behavior_constraint}
\end{equation}
In the sequence-form polytope \(\mathcal{Q}^\epsilon_i \subseteq \mathbb{R}_{\geq \boldsymbol{l}}^{\Sigma_i}\), the lower bound is \(\boldsymbol{l}[Ia] = \epsilon^{|Ia|}\), ensuring each action on the path is taken with at least probability \(\epsilon\).

\begin{definition}[\(\epsilon\)-Extensive-Form Perfect Equilibrium]
    In a perturbed EFG \(\Gamma^\epsilon = (\mathcal{Q}_1^\epsilon, \mathcal{Q}_2^\epsilon, \boldsymbol{U})\), a strategy profile \((\boldsymbol{q}_1^\epsilon, \boldsymbol{q}_2^\epsilon) \in \mathcal{Q}_1^\epsilon \times \mathcal{Q}_2^\epsilon\) is an \(\epsilon\)-EFPE if it is a Nash equilibrium of \(\Gamma^\epsilon\).
    \label{def:epsilon-EFPE}
\end{definition}

\begin{definition}[Extensive-Form Perfect Equilibrium]
    In an EFG \(\Gamma = (\mathcal{Q}_1, \mathcal{Q}_2, \boldsymbol{U})\), a strategy profile \((\boldsymbol{q}_1, \boldsymbol{q}_2) \in \mathcal{Q}_1 \times \mathcal{Q}_2\) is an EFPE if it is the limit of a sequence \(\{(\boldsymbol{q}_1^\epsilon, \boldsymbol{q}_2^\epsilon)\}_{\epsilon \to 0}\), where \((\boldsymbol{q}_1^\epsilon, \boldsymbol{q}_2^\epsilon)\) is an \(\epsilon\)-EFPE of \(\Gamma^\epsilon\).
    \label{def:EFPE}
\end{definition}

\subsection{Regret Matching and Counterfactual Regret Minimization}

Regret Matching (RM)~\cite{hart2000simple} and Regret Matching+ (RM+) are algorithms for computing Nash equilibria in NFGs. For player \(i \in \{1, 2\}\), each iteration computes:
\begin{align}
    \boldsymbol{v}_i^t &= \boldsymbol{U}_i \boldsymbol{x}_{-i}^t, \label{eq:v} \\
    \boldsymbol{r}_i^t &= \boldsymbol{v}_i^t - \langle \boldsymbol{v}_i^t, \boldsymbol{x}_i^t \rangle \mathbf{1}, \label{eq:r} \\
    \boldsymbol{R}_i^t &= \begin{cases}
        \boldsymbol{R}_i^{t-1} + \boldsymbol{r}_i^t, & \text{if RM}, \\
        [\boldsymbol{R}_i^{t-1} + \boldsymbol{r}_i^t]^+, & \text{if RM+},
    \end{cases} \label{eq:R} \\
    \boldsymbol{x}_i^{t+1} &= \begin{cases}
        \frac{[\boldsymbol{R}_i^t]^+}{\|[\boldsymbol{R}_i^t]^+\|_1}, & \text{if } \|[\boldsymbol{R}_i^t]^+\|_1 > 0, \\
        \frac{1}{|A_i|}, & \text{otherwise},
    \end{cases} \label{eq:x_rm}
\end{align}
where \(-i\) denotes the opponent, and \([\cdot]^+ = \max(\boldsymbol{0}, \cdot)\) takes the non-negative part of the vector.

Counterfactual Regret Minimization (CFR)~\cite{zinkevich2007regret} and CFR+~\cite{tammelin2014solving} extend RM and RM+ to EFGs by applying regret minimization independently at each information set. The counterfactual value at information set \(I \in \mathcal{I}_i\) for action \(a \in A(I)\) is computed recursively:
\begin{equation}
    \boldsymbol{v}_i^t[Ia] = \langle \boldsymbol{U}_i, \boldsymbol{q}_{-i}^t \rangle[Ia] + \sum_{I'a' \in \Sigma_i : pI' = Ia} \boldsymbol{x}_i^t[I'a'] \boldsymbol{v}_i^t[I'a'].
    \label{eq:v_cfr}
\end{equation}
CFR converges to a Nash equilibrium by Blackwell’s approachability theorem, with a worst-case regret bound of \(O(\sqrt{T})\) after \(T\) iterations~\cite{zinkevich2007regret}, though practical convergence is often faster.

\section{RTCFR for Solving EFPE in Perturbed Games}\label{sec:RTCFR_in_perturbed_game}

This section introduces an efficient regret minimization framework for computing \(\epsilon\)-Extensive-Form Perfect Equilibria (EFPE) in perturbed EFGs. Building on the work of \citet{farina2017regret}, we extend the Reward Transformation (RT) framework proposed by \citet{meng2023efficient} to develop \emph{Reward Transformation Counterfactual Regret Minimization (RTCFR)}, a novel method for approximating \(\epsilon\)-EFPE with guaranteed asymptotic last-iterate convergence.

\subsection{Regret Minimization in Perturbed Normal-Form Games}\label{subsec:perturbed_nfg}

In regret minimization (RM) dynamics, as defined in Equations~\eqref{eq:v}--\eqref{eq:x_rm}, the strategy space is typically a standard simplex (see Figure~\ref{fig:basis}, left). However, incorporating perturbation constraints, as specified in Equation~\eqref{eq:behavior_constraint}, introduces significant challenges. To address this, we adopt the approach of \citet{farina2017regret}, which preserves RM dynamics by mapping strategies to the perturbed space using a non-standard basis.

\begin{figure}[h]
    \centering
    \includegraphics[width=0.9\linewidth]{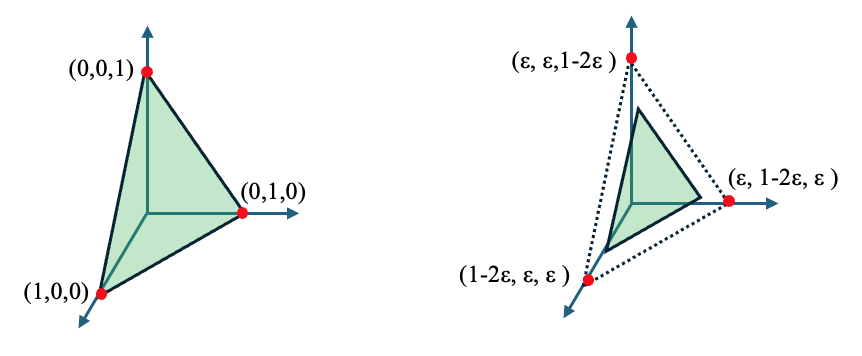}
    \caption{Unperturbed strategy polytope (left, standard simplex) and perturbed polytope (right) in a three-dimensional simplex.}
    \label{fig:basis}
    \vspace{0.5cm}
\end{figure}

Consider a NFG with a strategy space defined by an \(n\)-dimensional simplex \(\Delta^n\), with vertices \((\boldsymbol{b}_1, \dots, \boldsymbol{b}_n)\) forming the basis matrix \(\boldsymbol{B} = (\boldsymbol{b}_1 | \dots | \boldsymbol{b}_n)\). The strategy is computed as \(\boldsymbol{x} = \boldsymbol{B} \hat{\boldsymbol{x}}\), where \(\hat{\boldsymbol{x}} \in \Delta^n\) represents the coordinates. In an unperturbed game \(\Gamma\), the strategy space is a standard \(n\)-simplex, with each vertex \(\boldsymbol{b}_j^\top = (0, \dots, 1, \dots, 0)\), having a single non-zero entry at index \(j\) (Figure~\ref{fig:basis}, left), yielding \(\boldsymbol{x} = \hat{\boldsymbol{x}}\). For the perturbed game \(\Gamma^\epsilon\), the strategy space is a non-standard simplex \(\Delta_{\geq \epsilon}^n\), with \(\epsilon\)-perturbed vertices defined as:
\begin{equation}
    \boldsymbol{b}_j^\top = (\epsilon, \dots, \tau + \epsilon, \dots, \epsilon),
    \label{eq:basis}
\end{equation}
where the \(j\)-th entry is \(\tau + \epsilon\), and \(\tau = 1 - n\epsilon\) (Figure~\ref{fig:basis}, right).

Given the affine nature of the utility function, we have:
\begin{align}
    u_i(\boldsymbol{x}_i, \boldsymbol{x}_{-i}) &= \boldsymbol{x}_i^\top \boldsymbol{U}_i \boldsymbol{x}_{-i} = \langle \boldsymbol{B}_i \hat{\boldsymbol{x}}_i, \boldsymbol{U}_i \boldsymbol{B}_{-i} \hat{\boldsymbol{x}}_{-i} \rangle.
\end{align}
Define \(\hat{\boldsymbol{U}}_i = \boldsymbol{B}_i^\top \boldsymbol{U}_i \boldsymbol{B}_{-i}\), so that:
\begin{equation}
    u_i(\boldsymbol{x}_i, \boldsymbol{x}_{-i}) = \hat{\boldsymbol{x}}_i^\top \hat{\boldsymbol{U}}_i \hat{\boldsymbol{x}}_{-i} = \hat{u}_i(\hat{\boldsymbol{x}}_i, \hat{\boldsymbol{x}}_{-i}).
\end{equation}
This allows us to construct an equivalent transformed game \(\hat{\Gamma} = (\Delta^n, \Delta^m, \hat{\boldsymbol{U}})\) from \(\Gamma^\epsilon\), with the utility for the \(j\)-th vertex given by:
\begin{equation}
    \hat{\boldsymbol{v}}_i[j] = \boldsymbol{b}_j^\top \boldsymbol{v}_i,
\end{equation}
where \(\boldsymbol{v}_i = \boldsymbol{U}_i \boldsymbol{x}_{-i}\) for player \(i \in \{1, 2\}\). This ensures equivalence between the optimization objectives in \(\Gamma^\epsilon\) and \(\hat{\Gamma}\), as defined in Equation~\eqref{eq:bspp}.

The RM dynamics for the perturbed game \(\Gamma^\epsilon\) are:
\begin{align}
    \boldsymbol{v}_i^t &= \boldsymbol{U}_i \boldsymbol{x}_{-i}^t, \label{eq:p_v} \\
    \hat{\boldsymbol{v}}_i^t &= \boldsymbol{B}_i^\top \boldsymbol{v}_i^t, \label{eq:p_v'} \\
    \boldsymbol{r}_i^t &= \hat{\boldsymbol{v}}_i^t - \langle \hat{\boldsymbol{v}}_i^t, \hat{\boldsymbol{x}}_i^t \rangle \mathbf{1}, \label{eq:p_r} \\
    \boldsymbol{R}_i^t &= \begin{cases}
        \boldsymbol{R}_i^{t-1} + \boldsymbol{r}_i^t, & \text{if RM}, \\
        [\boldsymbol{R}_i^{t-1} + \boldsymbol{r}_i^t]^+, & \text{if RM+},
    \end{cases} \label{eq:p_R} \\
    \hat{\boldsymbol{x}}_i^{t+1} &= \begin{cases}
        \frac{[\boldsymbol{R}_i^t]^+}{\|[\boldsymbol{R}_i^t]^+\|_1}, & \text{if } \|[\boldsymbol{R}_i^t]^+\|_1 > 0, \\
        \frac{1}{|A_i|}, & \text{otherwise},
    \end{cases} \label{eq:p_x'} \\
    \boldsymbol{x}_i^{t+1} &= \boldsymbol{B}_i \hat{\boldsymbol{x}}_i^{t+1}. \label{eq:p_x}
\end{align}
Equation~\eqref{eq:p_v'} computes the perturbed action values in \(\hat{\Gamma}\), ensuring that the expected utility is invariant, i.e., \(\langle \hat{\boldsymbol{v}}_i^t, \hat{\boldsymbol{x}}_i^t \rangle = \langle \boldsymbol{v}_i^t, \boldsymbol{x}_i^t \rangle\) in Equation~\eqref{eq:p_r}. Equation~\eqref{eq:p_x} maps the strategy back to the original space, preserving RM dynamics under perturbation constraints.

\subsection{RTCFR for Perturbed Extensive-Form Games}\label{subsec:RTCFR_perturbed_EFG}

The Reward Transformation (RT) framework \cite{perolat2021poincare} achieves asymptotic last-iterate convergence by incorporating a regularization term, specifically the Euclidean distance between a reference strategy \(\boldsymbol{x}_i^r\) and the current strategy \(\boldsymbol{x}_i\) \cite{meng2023efficient, abe2022mutation}. This defines the \emph{Reward Transformation Bilinear Saddle-Point Problem (RT-BSPP)}:
\begin{equation}
    \max_{\boldsymbol{x}_1 \in \mathcal{X}_1} \min_{\boldsymbol{x}_2 \in \mathcal{X}_2} \boldsymbol{x}_1^\top \boldsymbol{U} \boldsymbol{x}_2 + \frac{\mu}{2} \|\boldsymbol{x}_1^r - \boldsymbol{x}_1\|_2^2 - \frac{\mu}{2} \|\boldsymbol{x}_2^r - \boldsymbol{x}_2\|_2^2,
    \label{eq:rt_bspp}
\end{equation}
where \(\mu > 0\) is the RT weight. For a reference strategy \(\boldsymbol{x}^r\), the RT-BSPP yields a saddle point \(\boldsymbol{x}^{*,r}\), which is biased relative to the NE \(\boldsymbol{x}^*\) of Equation~\eqref{eq:bspp}:
\begin{equation}
    \|\boldsymbol{x}^{*,r} - \boldsymbol{x}^*\|_2 \propto \mu \|\boldsymbol{x}^r - \boldsymbol{x}^*\|_2,
    \label{eq:reference_condition}
\end{equation}
as established by \citet{perolat2021poincare}. The RT framework ensures the convergence relation:
\begin{equation}
    \boldsymbol{x}^t \to \boldsymbol{x}^{*,r} \to \boldsymbol{x}^*.
    \label{eq:relation}
\end{equation}
By adaptively updating \(\boldsymbol{x}^r\) during training, \(\boldsymbol{x}^{*,r}\) converges to \(\boldsymbol{x}^*\), enabling asymptotic convergence to the NE.

We extend the RT framework to RM dynamics (Equations~\eqref{eq:p_v}--\eqref{eq:p_x}), including RM, RM+, and discounted RM \cite{brown2019solving}, collectively termed \emph{Reward Transformation Regret Minimization (RTRM)}. The value update in Equation~\eqref{eq:p_v} is modified as:
\begin{equation}
    \boldsymbol{v}_i^t = \boldsymbol{U}_i \boldsymbol{x}_{-i}^t + \mu (\boldsymbol{x}_i^r - \boldsymbol{x}_i^t).
    \label{eq:p_v_rtrm}
\end{equation}
To adaptively update \(\boldsymbol{x}^r\), we divide \(NT\) training iterations into \(N\) RT-BSPPs, each with \(T\) iterations. The final state of the \((n-1)\)-th RT-BSPP initializes the \(n\)-th:
\begin{equation}
    \boldsymbol{x}^{1,n} \gets \boldsymbol{x}^{T+1,n-1}, \quad \boldsymbol{R}^{1,n} \gets \boldsymbol{R}^{T,n-1}, \quad \boldsymbol{x}^{r,n} \gets \boldsymbol{x}^{T+1,n-1}.
    \label{eq:rt_init}
\end{equation}

We prove that RTRM satisfies Equation~\eqref{eq:relation}, achieving best-iterate convergence (\(\boldsymbol{x}^{t,n} \to \boldsymbol{x}^{*,n}\), where \(\boldsymbol{x}^{*,n} := \boldsymbol{x}^{*,r,n}\)) for each \(n\)-th RT-BSPP. Detailed proofs are provided in the Appendix.

\begin{theorem}[Best-Iterate Convergence of RTRM]\label{thm:best_iterate_converge_of_RTRM}
    Given a reference strategy \(\boldsymbol{x}^{r,n}\), RT weight \(\mu\), and perturbation \(\epsilon\), let \(\{\boldsymbol{x}^{t,n}\}\) be the strategy sequence produced by RTRM in the \(\epsilon\)-perturbed \(n\)-th RT-BSPP, with saddle point \(\boldsymbol{x}^{*,n} \in \mathcal{X}^{*,n}\). For any \(T \geq 1\), there exists \(t \in \{1, \dots, T\}\) such that:
    \[
        \|\boldsymbol{x}^{*,n} - \boldsymbol{x}^{t,n}\|_2 \leq O\left(\frac{1}{\sqrt{T}}\right).
    \]
\end{theorem}

By iteratively updating \(\boldsymbol{x}^{r,n}\) via Equation~\eqref{eq:rt_init}, the saddle point \(\boldsymbol{x}^{*,n}\) converges to the NE \(\boldsymbol{x}^*\), ensuring asymptotic last-iterate convergence to the NE of the \(\epsilon\)-perturbed game, as formalized in Theorem~\ref{thm:asymptotic_last_iterate_converge_of_RTRM}.

\begin{theorem}[Asymptotic Last-Iterate Convergence of RTRM]\label{thm:asymptotic_last_iterate_converge_of_RTRM}
    In an \(\epsilon\)-perturbed game \(\Gamma^\epsilon = (\mathcal{X}_1^\epsilon, \mathcal{X}_2^\epsilon, \boldsymbol{U})\), let \(\{\boldsymbol{x}^{*,1}, \dots, \boldsymbol{x}^{*,n}\}\) be the saddle-point sequence across \(n\) RT-BSPPs. Then, \(\boldsymbol{x}^{*,n}\) is bounded, converges to the NE \(\boldsymbol{x}^*\) of \(\Gamma^\epsilon\), and the strategy \(\boldsymbol{x}^{t,n}\) produced by RTRM asymptotically converges to \(\boldsymbol{x}^*\).
\end{theorem}

We further extend RTRM to Counterfactual Regret Minimization (CFR) for perturbed EFGs, performing local RTRM optimization at each information set in a bottom-up manner, termed \emph{Reward Transformation Counterfactual Regret Minimization (RTCFR)}. For each action \(a\) at information set \(I \in \mathcal{I}_i\), the RT counterfactual value is:
\begin{equation}
    \tilde{\boldsymbol{v}}_i^t[I,a] = \boldsymbol{v}_i^t[I,a] + \mu (\boldsymbol{x}_i^r[I,a] - \boldsymbol{x}_i^t[I,a]).
    \label{eq:v_rtcfr}
\end{equation}

We establish best-iterate convergence of RTCFR to the saddle point in the \(\epsilon\)-perturbed \(n\)-th RT-BSPP using sequence-form strategies:

\begin{theorem}[Best-Iterate Convergence of RTCFR]\label{thm:best_iterate_convergence_of_RTCFR}
    Given a reference strategy \(\boldsymbol{q}^{r,n}\), RT weight \(\mu\), and perturbation \(\epsilon\), let \(\{\boldsymbol{q}^{t,n}\}\) be the strategy sequence produced by RTCFR in the \(\epsilon\)-perturbed \(n\)-th RT-BSPP, with saddle point \(\boldsymbol{q}^{*,n} \in \mathcal{Q}^{*,n}\). For any \(T \geq 1\), there exists \(t \in \{1, \dots, T\}\) such that:
    \[
        \|\boldsymbol{q}^{*,n} - \boldsymbol{q}^{t,n}\|_2 \leq O\left(\frac{1}{\sqrt{T}}\right).
    \]
\end{theorem}

By updating \(\boldsymbol{q}^{r,n} \gets \boldsymbol{q}^{T+1,n-1}\), RTCFR achieves asymptotic last-iterate convergence to an \(\epsilon\)-EFPE:

\begin{theorem}[Asymptotic Last-Iterate Convergence of RTCFR]\label{thm:asymptotic_last_iterate_converge_of_RTCFR}
    In an \(\epsilon\)-perturbed EFG \(\Gamma^\epsilon = (\mathcal{Q}_1^\epsilon, \mathcal{Q}_2^\epsilon, \boldsymbol{U})\), let \(\{\boldsymbol{q}^{*,1}, \dots, \boldsymbol{q}^{*,n}\}\) be the saddle-point sequence across \(n\) RT-BSPPs. Then, \(\boldsymbol{q}^{*,n}\) is bounded, converges to the \(\epsilon\)-EFPE \(\boldsymbol{q}^*\) of \(\Gamma^\epsilon\), and the strategy \(\boldsymbol{q}^{t,n}\) produced by RTCFR asymptotically converges to \(\boldsymbol{q}^*\).
\end{theorem}

For Theorem~\ref{thm:asymptotic_last_iterate_converge_of_RTCFR}, a strict theoretical guarantee requires \(T \to \infty\) in each RT-BSPP to ensure \(\boldsymbol{q}^{r,n} = \boldsymbol{q}^{T+1,n-1} = \boldsymbol{q}^{*,n-1}\). In practice, a carefully tuned finite \(T\) yields a reference strategy sequence \(\{\boldsymbol{q}^{r,n} \gets \boldsymbol{q}^{T+1,n-1}\}\) that effectively converges to the \(\epsilon\)-EFPE \(\boldsymbol{q}^*\), satisfying Equation~\eqref{eq:reference_condition} and achieving best-iterate convergence per Theorem~\ref{thm:best_iterate_convergence_of_RTCFR}, thus fulfilling Theorem~\ref{thm:asymptotic_last_iterate_converge_of_RTCFR}.

\section{Adaptive Perturbation for RTCFR in Solving EFPE}\label{sec:adaptive_setting}

This section investigates the limitations of fixed perturbations in CFR dynamics and proposes an adaptive perturbation method for RTCFR. Guided by a novel \emph{Information Set Nash Equilibrium (ISNE)} metric, our approach dynamically adjusts perturbations to balance convergence speed and equilibrium accuracy in computing EFPE.

\subsection{Challenges of Fixed Perturbations}

Although fixed perturbations ensure reachability of all information sets, small perturbation values hinder strategy updates in subgames with low reach probabilities, compromising stability. In games with maximum depth \(d\), setting a minimum action probability \(\epsilon\) can result in branches with arrival probabilities as low as \(\epsilon^d\), complicating optimization in these subgames.

In RTCFR dynamics (Equations~\eqref{eq:v_rtcfr}, \eqref{eq:p_v'}--\eqref{eq:p_x}), for an information set \(I \in \mathcal{I}_i\) and action \(a \in A(I)\), regret matching assigns strategies proportional to positive cumulative regret: \(\boldsymbol{x}_i^{t+1}[Ia] \propto [\boldsymbol{R}_i^t[Ia]]^+ / \|[\boldsymbol{R}_i^t(I)]^+\|_1\). When the opponent’s reach probability is low (\(q_{-i}^t(I) \to 0\)), counterfactual values \(\boldsymbol{v}_i^t[Ia]\) diminish, leading to small regrets \(\boldsymbol{r}_i^t[Ia]\). This impedes updates to \(\boldsymbol{R}_i^t[Ia]\), causing subgame strategies to deviate significantly from the EFPE.

Fixed perturbations create a trade-off between NE approximation accuracy and the convergence rate of NE refinements, as noted by \citet{farina2017regret, kroer2017smoothing}. Large perturbations accelerate refinement but risk premature stagnation in NE convergence, while small perturbations improve accuracy at the cost of slower convergence. This motivates an adaptive perturbation strategy, which we develop within the RTCFR framework below.

\subsection{Adaptive Perturbation Setting}\label{subsec:adaptive_perturbation}

The EFPE requires equilibrium conditions at every information set, including those with low reach probabilities. To evaluate strategy stability, we introduce the \emph{Information Set Regret} metric:

\begin{definition}[Information Set Regret]\label{def:info_set_regret}
For an information set \(I \in \mathcal{I}_i\) of player \(i \in \{1, 2\}\), given action values \(\boldsymbol{v}_i'(I) \in \mathbb{R}^{|A(I)|}\) assuming all players reach \(I\), the information set regret is:
\[
    \boldsymbol{r}_i'(I) = \boldsymbol{v}_i'(I) - \langle \boldsymbol{v}_i'(I), \boldsymbol{x}_i(I) \rangle \mathbf{1}.
\]
\end{definition}

Unlike counterfactual regret, which weights regrets by opponent reach probabilities, information set regret assumes the opponent reaches \(I\) with probability 1. Using Bayes' rule \cite{srinivasan2018actor}, it relates to counterfactual regret as:
\[
    \boldsymbol{r}_i'(I) = \frac{\boldsymbol{r}_i(I)}{\sum_{h \in I} q_{-i}(h)},
\]
where \(\boldsymbol{r}_i(I)\) is the counterfactual regret at \(I\), and \(q_{-i}(h)\) is the opponent’s reach probability for state \(h \in I\). Strategy stability at \(I\) is proportional to the maximum regret \(\max(\boldsymbol{r}_i'(I))\), with an equilibrium achieved when \(\max(\boldsymbol{r}_i'(I)) = 0\).

To quantify proximity to an exact \(\epsilon\)-EFPE, we define the \(\delta\)-Information Set Nash Equilibrium (\(\delta\)-ISNE):

\begin{definition}[\(\delta\)-ISNE]\label{def:delta_isne}
In an \(\epsilon\)-perturbed EFG \(\Gamma^\epsilon = (\mathcal{Q}_1^\epsilon, \mathcal{Q}_2^\epsilon, \boldsymbol{U})\), let \(r^{\max}(\boldsymbol{q}) = \max_{i \in \{1, 2\}} \max_{I \in \mathcal{I}_i} \max(\boldsymbol{r}_i'(I))\) be the maximum information set regret for a strategy profile \(\boldsymbol{q} \in \mathcal{Q}^\epsilon\). A strategy \(\boldsymbol{q}\) is a \(\delta\)-ISNE if:
\[
    r^{\max}(\boldsymbol{q}) \leq \delta.
\]
\end{definition}

Typically, the exploitability \(\text{Exp}(\boldsymbol{q}) \ll \delta |\mathcal{I}_1 \cup \mathcal{I}_2|\) in the \(\epsilon\)-perturbed game \(\Gamma^\epsilon\), making the \(\delta\)-ISNE a robust approximation of an \(\epsilon\)-EFPE both locally and globally. The \(\delta\)-ISNE captures subgame perfection in imperfect-information games, analogous to subgame-perfect equilibria in perfect-information games \cite{selten1965spieltheoretische}, conditioned on state distributions within information sets. Reducing \(\delta\) yields a more precise \(\epsilon\)-EFPE, with an exact \(\epsilon\)-EFPE achieved as \(\delta \to 0\). By Theorem~\ref{thm:asymptotic_last_iterate_converge_of_RTCFR}, the maximum information set regret \(r^{\max}(\boldsymbol{q}^{t,n})\) vanishes as \(n, T \to \infty\), ensuring that a \(\delta\)-ISNE converges to an \(\epsilon\)-EFPE as \(\delta \to 0\). Thus, the \(\delta\)-ISNE serves as an effective metric for computing \(\epsilon\)-EFPEs.

\begin{algorithm}[tb]
   \caption{RTCFR with Adaptive Perturbation for Perturbed EFGs}
   \label{alg:adaptive_perturbation}
\begin{algorithmic}[1]
   \STATE {\bfseries Input:} RT weight \(\mu\), iterations \(T\), initial perturbation \(\epsilon\), ISNE threshold \(\delta\), decay factor \(\gamma \in (0, 1)\)
   \STATE Initialize \(\boldsymbol{x}^{1,1}(I) \gets \frac{\mathbf{1}}{|A(I)|}\), \(\boldsymbol{R}^{0,1}(I) \gets \mathbf{0}\), \(\forall I \in \mathcal{I}\)
   \FOR{\(n = 1\) {\bfseries to} \(N\)}
       \STATE \(\boldsymbol{x}^{1,n} \gets \boldsymbol{x}^{T+1,n-1}\), \(\boldsymbol{R}^{1,n} \gets \boldsymbol{R}^{T,n-1}\), \(\boldsymbol{x}^{r,n} \gets \boldsymbol{x}^{T+1,n-1}\)
       \IF{\(r^{\text{max}}(\boldsymbol{x}^{1,n}) < \delta\)}
           \STATE \(\epsilon \gets \epsilon \cdot \gamma, \quad \delta \gets \delta \cdot \gamma\)
           \STATE Update basis matrix \(\boldsymbol{B}(I)\) using \(\epsilon\), \(\forall I \in \mathcal{I}\)
       \ENDIF
       \FOR{\(t = 1\) {\bfseries to} \(T\)}
           \FOR{player \(i \in \{1, 2\}\)}
               \STATE Compute counterfactual value: 
                     \(\boldsymbol{v}_i^{t,n} \gets \langle \boldsymbol{U}_i, \boldsymbol{q}_{-i}^{t,n} \rangle \)
               \FOR{each information set \(I \in \mathcal{I}_i\) (bottom-up)}
                   \STATE compute RT counterfactual value: \\
                            \(\tilde{\boldsymbol{v}}_i^{t,n}=\boldsymbol{v}_i^{t,n}(I)+\mu(\boldsymbol{x}_i^{r,n}-\boldsymbol{x}_i^{t,n})\)
                   \STATE Compute expected value: 
                         \(u(I) \gets \langle \tilde{\boldsymbol{v}}_i^{t,n}(I), \boldsymbol{x}_i^{t,n}(I) \rangle\)
                   \STATE Compute regret in perturbed space: \\
                         \(\boldsymbol{r}_i^{t,n}(I) \gets \boldsymbol{B}_i^\top(I) \tilde{\boldsymbol{v}}_i^{t,n}(I) - u(I) \mathbf{1}\)
                   \STATE Update cumulative regret \(\boldsymbol{R}_i^{t,n}(I)\) and strategy \(\boldsymbol{\hat{x}}_i^{t+1,n}(I)\) using Equations~\eqref{eq:p_R}, \eqref{eq:p_x'}
                   \STATE Map to original space: 
                         \(\boldsymbol{x}_i^{t+1,n}(I) \gets \boldsymbol{B}_i(I) \boldsymbol{\hat{x}}_i^{t+1,n}(I)\)
                   \STATE Update parent sequence counterfactual value: \\
                         \(\boldsymbol{v}_i^{t,n}[pI] \gets \boldsymbol{v}_i^{t,n}[pI] + \langle \boldsymbol{v}_i^{t,n}(I),\boldsymbol{x}_i^{t,n}(I)\rangle\)
               \ENDFOR
           \ENDFOR
       \ENDFOR
   \ENDFOR
   \STATE {\bfseries Output:} \(\boldsymbol{x}^{T+1,N}\)
\end{algorithmic}
\end{algorithm}

We propose an adaptive perturbation method for RTCFR, detailed in Algorithm~\ref{alg:adaptive_perturbation}, which uses \(\delta\)-ISNE as a threshold to dynamically reduce the perturbation \(\epsilon\). The training process consists of \(N\) RT-BSPPs, each with \(T\) iterations. Within each RT-BSPP, we monitor \(r^{\max}(\boldsymbol{q}^{t,n})\) and reduce \(\epsilon\) and \(\delta\) when \(r^{\max}(\boldsymbol{q}^{t,n}) \leq \delta\), promoting balanced convergence across information sets (lines 5--8). By Theorem~\ref{thm:best_iterate_convergence_of_RTCFR}, this adaptive setting preserves RTCFR’s convergence to the saddle point \(\boldsymbol{q}^{*,n}\) of each RT-BSPP, despite potential early oscillations, ensuring asymptotic convergence to the EFPE \(\boldsymbol{q}^*\) as \(\epsilon \to 0\) (lines 9--22).

\section{Experiments}\label{sec:experiments}

\begin{figure*}[t]
    \centering
    \begin{subfigure}[b]{0.245\linewidth}
        \centering
        \includegraphics[width=\linewidth, height=0.12\textheight, keepaspectratio]{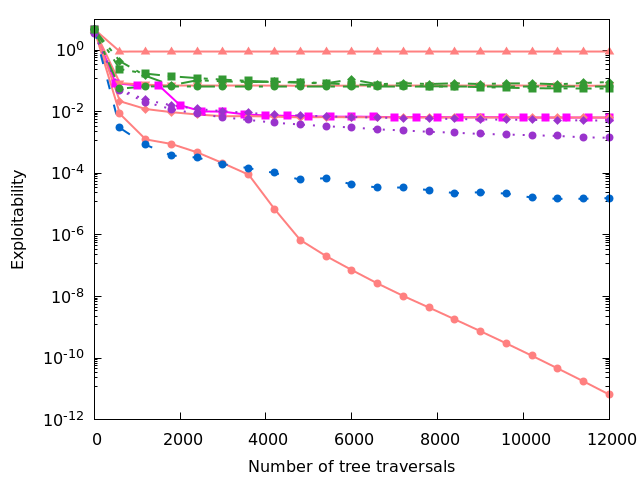}
        \caption{Leduc Poker (3) -- Exp.}
        \label{fig:leduc-3-exp}
    \end{subfigure}
    \hfill
    \begin{subfigure}[b]{0.245\linewidth}
        \centering
        \includegraphics[width=\linewidth, height=0.12\textheight, keepaspectratio]{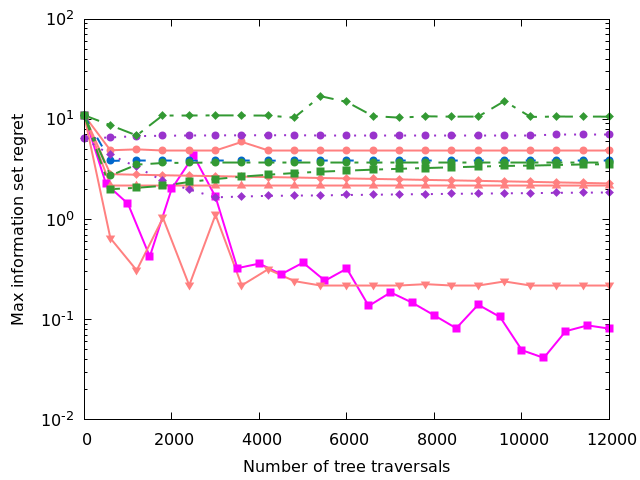}
        \caption{Leduc Poker (3) -- Max Regret}
        \label{fig:leduc-3-max-regret}
    \end{subfigure}
    \hfill
    \begin{subfigure}[b]{0.245\linewidth}
        \centering
        \includegraphics[width=\linewidth, height=0.12\textheight, keepaspectratio]{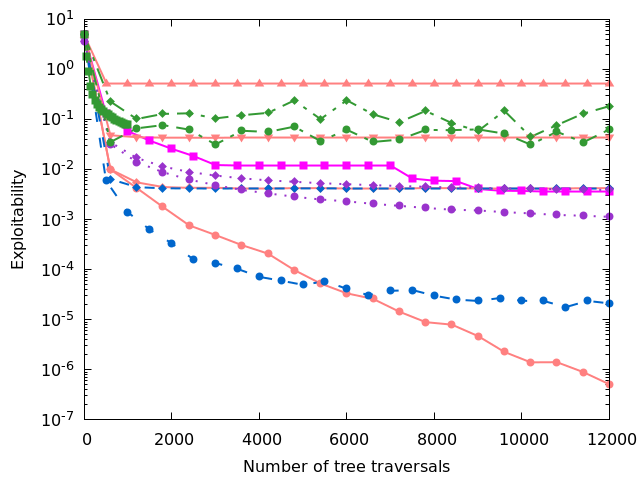}
        \caption{Leduc Poker (5) -- Exp.}
        \label{fig:l5-exp}
    \end{subfigure}
    \hfill
    \begin{subfigure}[b]{0.245\linewidth}
        \centering
        \includegraphics[width=\linewidth, height=0.12\textheight, keepaspectratio]{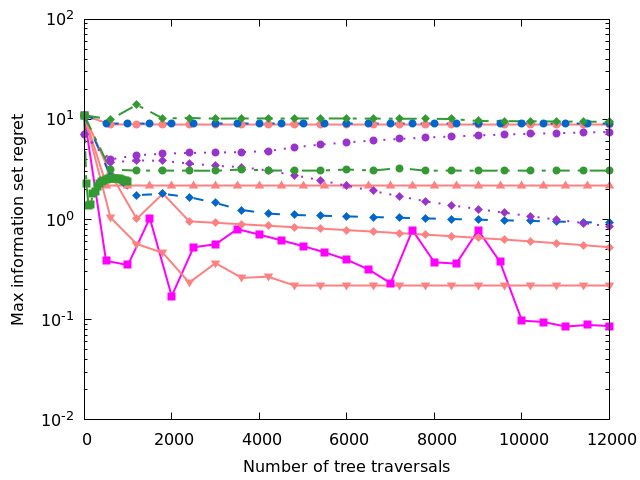}
        \caption{Leduc Poker (5) -- Max Regret}
        \label{fig:l5-max-regret}
    \end{subfigure}

    \vspace{0.3cm} 

    \begin{subfigure}[b]{0.245\linewidth}
        \centering
        \includegraphics[width=\linewidth, height=0.12\textheight, keepaspectratio]{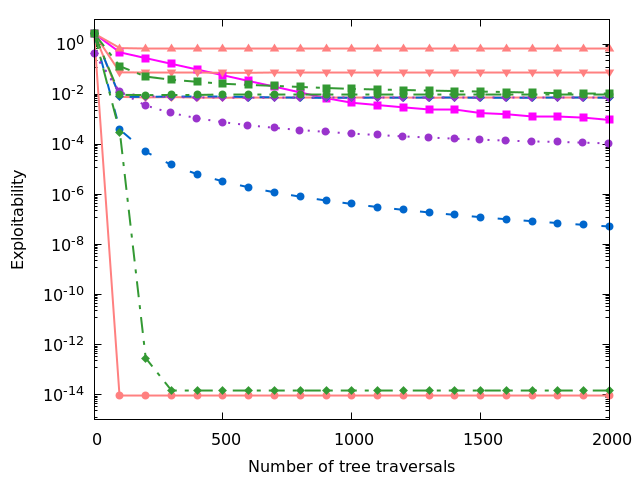}
        \caption{Goofspiel (3) -- Exp.}
        \label{fig:goofspiel3-exp}
    \end{subfigure}
    \hfill
    \begin{subfigure}[b]{0.245\linewidth}
        \centering
        \includegraphics[width=\linewidth, height=0.12\textheight, keepaspectratio]{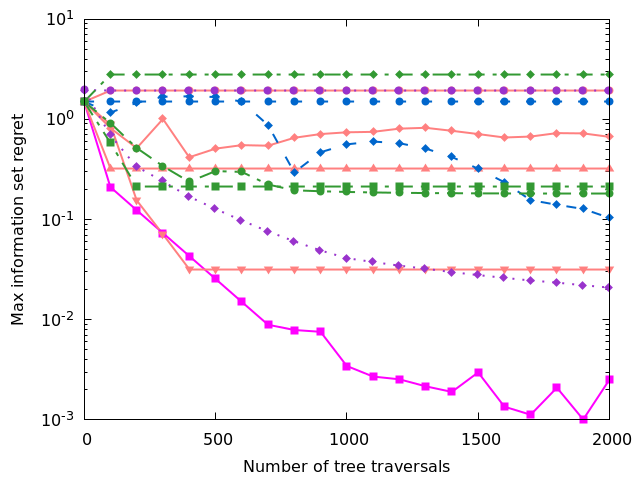}
        \caption{Goofspiel (3) -- Max Regret}
        \label{fig:goofspiel3-max-regret}
    \end{subfigure}
    \hfill
    \begin{subfigure}[b]{0.245\linewidth}
        \centering
        \includegraphics[width=\linewidth, height=0.12\textheight, keepaspectratio]{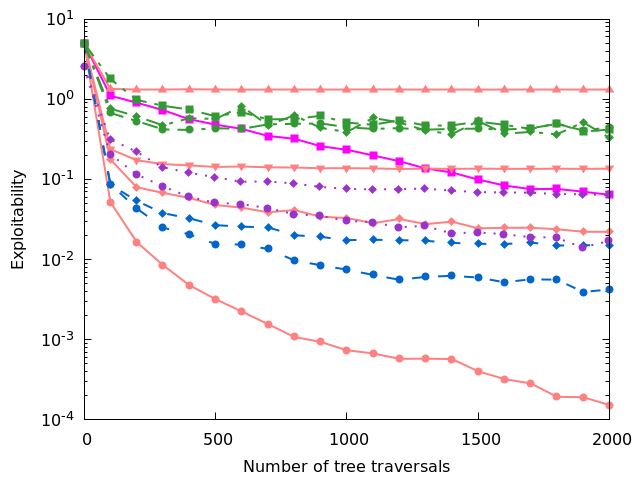}
        \caption{Goofspiel (4) -- Exp.}
        \label{fig:goofspiel4-exp}
    \end{subfigure}
    \hfill
    \begin{subfigure}[b]{0.245\linewidth}
        \centering
        \includegraphics[width=\linewidth, height=0.12\textheight, keepaspectratio]{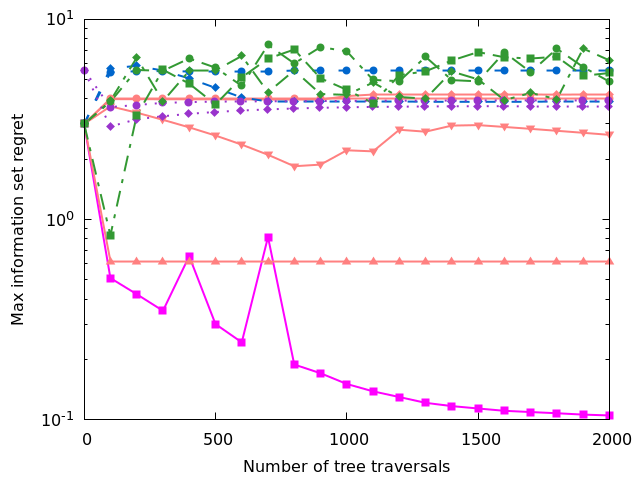}
        \caption{Goofspiel (4) -- Max Regret}
        \label{fig:goofspiel4-max-regret}
    \end{subfigure}

    \vspace{0.3cm} 

    \begin{subfigure}[b]{0.245\linewidth}
        \centering
        \includegraphics[width=\linewidth, height=0.12\textheight, keepaspectratio]{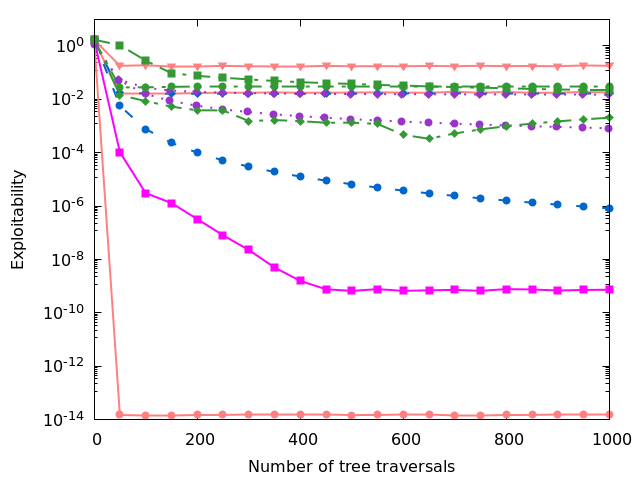}
        \caption{Liar's Dice (5) -- Exp.}
        \label{fig:liars_dice5-exp}
    \end{subfigure}
    \hfill
    \begin{subfigure}[b]{0.245\linewidth}
        \centering
        \includegraphics[width=\linewidth, height=0.12\textheight, keepaspectratio]{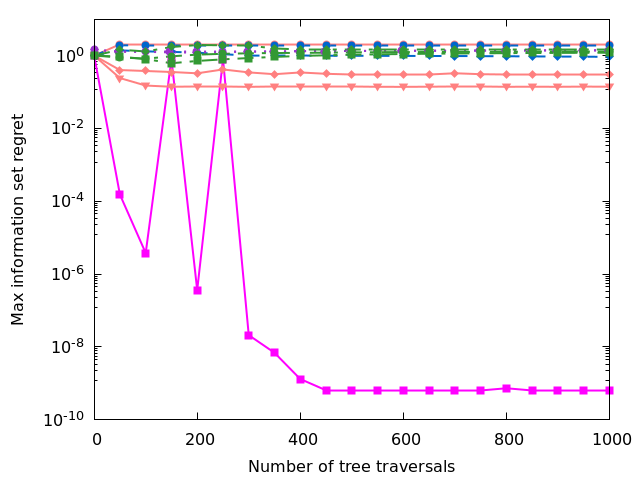}
        \caption{Liar's Dice (5) -- Max Regret}
        \label{fig:liars_dice5-max-regret}
    \end{subfigure}
    \hfill
    \begin{subfigure}[b]{0.245\linewidth}
        \centering
        \includegraphics[width=\linewidth, height=0.12\textheight, keepaspectratio]{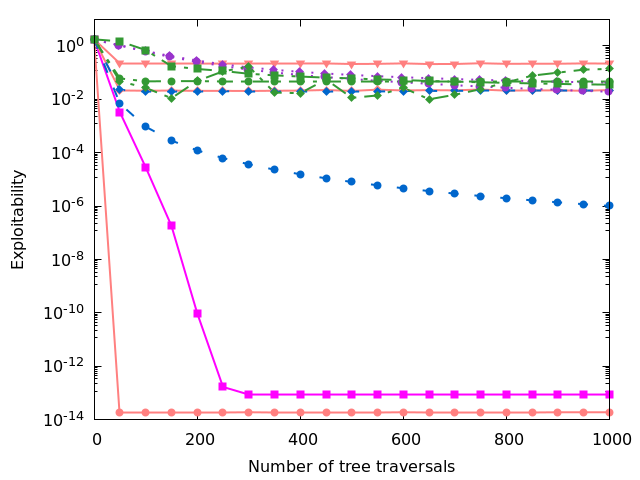}
        \caption{Liar's Dice (6) -- Exp.}
        \label{fig:liars_dice6-exp}
    \end{subfigure}
    \hfill
    \begin{subfigure}[b]{0.245\linewidth}
        \centering
        \includegraphics[width=\linewidth, height=0.12\textheight, keepaspectratio]{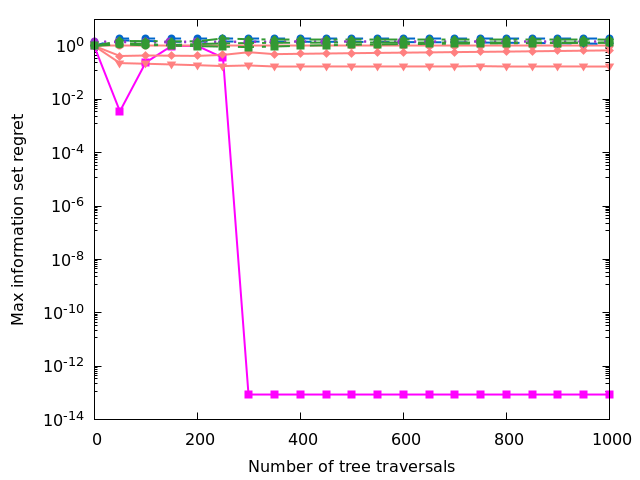}
        \caption{Liar's Dice (6) -- Max Regret}
        \label{fig:liars_dice6-max-regret}
    \end{subfigure}

    \vspace{0.3cm} 

    \begin{subfigure}[b]{0.245\linewidth}
        \centering
        \includegraphics[width=\linewidth, height=0.12\textheight, keepaspectratio]{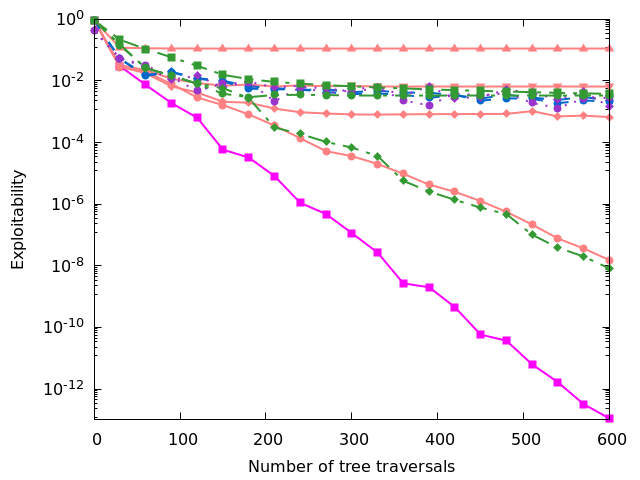}
        \caption{Kuhn Poker (3) -- Exp.}
        \label{fig:kuhn3-exp}
    \end{subfigure}
    \hfill
    \begin{subfigure}[b]{0.245\linewidth}
        \centering
        \includegraphics[width=\linewidth, height=0.12\textheight, keepaspectratio]{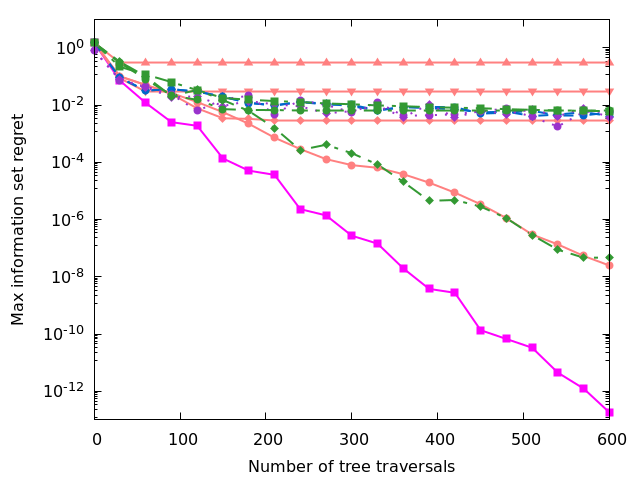}
        \caption{Kuhn Poker (3) -- Max Regret}
        \label{fig:kuhn3-max-regret}
    \end{subfigure}
    \hfill
    \begin{subfigure}[b]{0.48\linewidth}
        \centering
        \includegraphics[width=\linewidth, height=0.12\textheight, keepaspectratio]{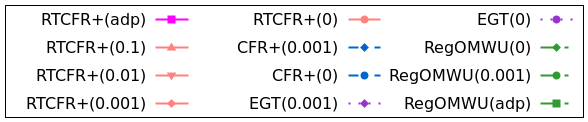}
        \vspace{0.1cm}
        \caption{Legend}
        \label{fig:legend}
    \end{subfigure}
    \vspace{0.3cm}
    \caption{Performance comparison for NE (exploitability) and EFPE (maximum information set regret) computation. ``Game (\(n\))'' denotes variants with \(n\)-rank cards or dice.}
    \label{fig:games-comparison}
    \vspace{0.3cm}
\end{figure*}

We evaluate our proposed methods, RTCFR+ with fixed perturbation (\(\epsilon\)) (Section~\ref{sec:RTCFR_in_perturbed_game}) and RTCFR+ with adaptive perturbation (Section~\ref{subsec:adaptive_perturbation}), both implemented locally with RTRM+, on four two-player zero-sum benchmark games: Kuhn Poker, Leduc Poker, Goofspiel, and Liar’s Dice, with varying rank variants.\footnote{Our code is available at \url{https://github.com/cat1994/RM_EFPE_experiments}.} Experiments were conducted on a PC with a 24-core CPU (up to 5.8 GHz) and 32 GB of memory. Table~\ref{tab:game_info} summarizes the game sizes, including the number of information sets, sequences, and terminal nodes (payoff leaves).
\begin{table}[t]
    \centering
    \caption{Game sizes used in experiments}
    \begin{tabular}{lrrr}
        \toprule
        Game Instance & Information Sets & Sequences & Leaves \\
        \midrule
        Kuhn Poker (3)   & 12    & 26    & 30     \\
        Leduc Poker (3)  & 288   & 674   & 1116   \\
        Leduc Poker (5)  & 780   & 1822  & 5500   \\
        Goofspiel (3)    & 546   & 668   & 216    \\
        Goofspiel (4)    & 34952 & 42658 & 13824  \\
        Liar's Dice (5)  & 5120  & 10232 & 25575  \\
        Liar's Dice (6)  & 24576 & 49142 & 147420 \\
        \bottomrule
    \end{tabular}
    \label{tab:game_info}
\end{table}

We compare RTCFR+ against state-of-the-art algorithms for NE and EFPE computation: unperturbed methods CFR+~\cite{zinkevich2007regret}, EGT~\cite{nesterov2005excessive}, and RegOMWU~\cite{liu2022power} (denoted with (0)), and their perturbed variants CFR+ (\(\epsilon\))~\cite{farina2017regret}, EGT (\(\epsilon\))~\cite{kroer2017smoothing}, RegOMWU (\(\epsilon\))~\cite{liu2022power}, and RegOMWU (adp)~\cite{bernasconi2024learning}. To handle zero reach probabilities in information set regret calculations, we set \(\epsilon = 10^{-15}\) for unperturbed methods. All algorithms use alternating updates for efficiency and evaluate the last strategy, except CFR+, which relies on average-iterate convergence and uses quadratic averaging: \(\frac{6}{T(T+1)(2T+1)} \sum_{t=1}^T t^2 \boldsymbol{q}^t\) for sequence-form strategies \(\{\boldsymbol{q}^1, \dots, \boldsymbol{q}^T\}\).

For fixed-perturbation methods, we test \(\epsilon \in \{0.1, 0.01, 0.001\}\) for RTCFR+ and \(\epsilon = 0.001\) for CFR+, EGT, and RegOMWU. RTCFR+ (adp) achieves a final perturbation \(\epsilon \leq 0.001\) across all games. Hyperparameters for RTCFR+ (adp) are listed in Table~\ref{tab:par_in_EFG}, tuned via grid search over limited iterations. For RTCFR+, we search \((T, \mu) \in \{1, 5, 10, 20, \dots\} \times \{0.1, 0.01, 0.001, \dots\}\) for NE approximation and adaptive parameters \((\epsilon^0, \delta, \gamma) \in \{0.1\} \times \{1, 0.5, 0.1\} \times \{0.95, 0.9, 0.85, \dots\}\). In Liar’s Dice, RTCFR uses \(T = 1\), omitting the RT term (\(\mu (\boldsymbol{x}^{r,n} - \boldsymbol{x}^{t,n}) = 0\)) as \(\boldsymbol{x}^{r,n} = \boldsymbol{x}^{t,n}\), aligning with CFR+ and converging to the NE in 20 iterations. 
Additional details are provided in the Appendix.

\begin{table}[t]
    \centering
    \caption{Hyperparameters for RTCFR+ (adp) in EFGs}
    \begin{tabular}{lrllll}
        \toprule
        Game & \(T\) & \(\mu\) & \(\epsilon^0\) & \(\delta\) & \(\gamma\) \\
        \midrule
        Kuhn Poker (3)   & 5   & 0.01   & 0.1 & 1   & 0.5  \\
        Leduc Poker (3)  & 200 & 0.0001 & 0.01 & 0.02 & 0.1  \\
        Leduc Poker (5)  & 200 & 0.0001 & 0.1 & 0.5  & 0.5  \\
        Goofspiel (3)    & 20  & 0.001  & 0.1 & 0.5  & 0.95 \\
        Goofspiel (4)    & 30  & 0.001  & 0.1 & 0.5  & 0.9  \\
        Liar’s Dice (5)  & 1   & --     & 0.1 & 0.5  & 0.5  \\
        Liar’s Dice (6)  & 1   & --     & 0.1 & 0.5  & 0.5  \\
        \bottomrule
    \end{tabular}
    \label{tab:par_in_EFG}
\end{table}

Figure~\ref{fig:games-comparison} illustrates the results for NE computation (exploitability) and EFPE computation (maximum information set regret). The x-axis denotes game tree traversals (\(T'\)). Due to higher traversal requirements, EGT incurs three times more traversals per iteration, and CFR+ incurs twice the cost, whereas RTCFR+ requires only one traversal per iteration. For RTCFR, \(N\) RT-BSPPs, each with \(T\) traversals, result in \(T' = NT\).

RTCFR+ achieves the lowest exploitability and fastest NE convergence, approaching \(O(1/T)\) in practice, surpassing its theoretical bound. This is due to effective reference strategy updates within tuned RT-BSPP traversals. In fixed-perturbation settings, baseline algorithms struggle to reduce maximum information set regret, except in the simpler Kuhn Poker. For RTCFR+ (\(\epsilon\)) with large \(\epsilon = 0.1\), regret decreases rapidly but exploitability remains high. With small \(\epsilon = 0.001\), regret reduction is slower, with no progress in Goofspiel (3) (Figure~\ref{fig:goofspiel3-max-regret}). The trade-off setting (\(\epsilon = 0.01\)) performs best overall, confirming the challenges noted by \citet{farina2017regret, kroer2017smoothing}.

RTCFR+ (adp) effectively balances NE and EFPE convergence across all games, achieving faster rates and tighter bounds. For example, in Leduc Poker (5) (Figures~\ref{fig:l5-exp}, \ref{fig:l5-max-regret}), RTCFR+ (adp) starts with \(\epsilon = 0.1\), reducing it by \(\gamma = 0.5\) upon reaching the \(\delta\)-ISNE threshold, leading to rapid exploitability reduction. Fluctuations in maximum information set regret reflect RTCFR+’s asymptotic last-iterate convergence, which is less smooth than averaging methods. Perturbation reductions shift the RT-BSPP saddle point, requiring RTCFR to converge to the new \(\epsilon\)-EFPE.

The RTCFR algorithm demonstrates high efficiency in Goofspiel and Liar’s Dice, requiring fewer traversals compared to traditional methods. In particular, for Liar’s Dice (6), despite being the largest game considered, RTCFR+ converges rapidly. This is achieved by reducing to the original CFR+ algorithm when setting \( T = 1 \), which nullifies the RT-term (\( \mu (\boldsymbol{x}^{r,n} - \boldsymbol{x}^{t,n}) = 0 \)). This behavior is attributed to the game possessing a strict NE, where the original CFR also achieves last-iterate convergence \cite{cai2023last}. Our adaptive method mitigates this, achieving exploitability comparable to unperturbed settings and an ISNE with regret below \( 10^{-10} \).

\section{Conclusion and Future Work}\label{sec:conclusion}

This paper presents an efficient last-iterate convergence framework for computing Extensive-Form Perfect Equilibria (EFPE) in Extensive-Form Games (EFGs). We propose \emph{Reward Transformation Counterfactual Regret Minimization (RTCFR)}, which extends the Reward Transformation framework with a perturbed game formulation and proves its asymptotic last-iterate convergence to an \(\epsilon\)-EFPE. To address challenges in low-reach-probability information sets, we introduce an adaptive perturbation strategy guided by the novel \emph{Information Set Nash Equilibrium (ISNE)} metric. This approach dynamically adjusts perturbations to balance Nash Equilibrium (NE) approximation accuracy and EFPE convergence speed. Experimental results across benchmark games demonstrate that RTCFR+ outperforms state-of-the-art methods, including CFR+, EGT, and RegOMWU, in both NE and EFPE computation.

Our findings underscore the potential of last-iterate convergence and adaptive perturbations for advancing game-theoretic solutions in large-scale and real-world applications. By moving beyond average-iterate methods, RTCFR can integrate with deep learning frameworks, as seen in Deep Nash~\cite{perolat2022mastering}. Future work could explore deep reinforcement learning, such as actor-critic methods, to model irrational behaviors in constrained environments, offering theoretical insights and practical benefits. These directions merit further investigation.



\begin{ack}
This research was funded by Guangdong Provincial Key Laboratory of Novel Security Intelligence Technologies (2022B1212010005), National Natural Science Foundation of China (No. 62376073), Natural Science Foundation of Guang-dong (No. 2024A1515030024), the Colleges and Universities Stable Support Project of Shenzhen (No. GXWD20220811173149002), Shenzhen Science and Technology Program under Grant (No. KJZD20230923114213027).
\end{ack}



\bibliography{m2455}

\begin{thebibliography}{38}
\providecommand{\natexlab}[1]{#1}
\providecommand{\url}[1]{\texttt{#1}}
\expandafter\ifx\csname urlstyle\endcsname\relax
  \providecommand{\doi}[1]{doi: #1}\else
  \providecommand{\doi}{doi: \begingroup \urlstyle{rm}\Url}\fi

\bibitem[Abe et~al.(2022)Abe, Sakamoto, and Iwasaki]{abe2022mutation}
K.~Abe, M.~Sakamoto, and A.~Iwasaki.
\newblock Mutation-driven follow the regularized leader for last-iterate convergence in zero-sum games.
\newblock In \emph{Uncertainty in Artificial Intelligence}, pages 1--10, 2022.

\bibitem[Abe et~al.(2023)Abe, Ariu, Sakamoto, Toyoshima, and Iwasaki]{abe2022last}
K.~Abe, K.~Ariu, M.~Sakamoto, K.~Toyoshima, and A.~Iwasaki.
\newblock Last-iterate convergence with full and noisy feedback in two-player zero-sum games.
\newblock In \emph{International Conference on Artificial Intelligence and Statistics}, volume 206, pages 7999--8028, 2023.

\bibitem[Anagnostides et~al.(2022)Anagnostides, Panageas, Farina, and Sandholm]{anagnostides2022last}
I.~Anagnostides, I.~Panageas, G.~Farina, and T.~Sandholm.
\newblock On last-iterate convergence beyond zero-sum games.
\newblock In \emph{International Conference on Machine Learning}, pages 536--581, 2022.

\bibitem[Bernasconi et~al.(2024)Bernasconi, Marchesi, and Trov{\`o}]{bernasconi2024learning}
M.~Bernasconi, A.~Marchesi, and F.~Trov{\`o}.
\newblock Learning extensive-form perfect equilibria in two-player zero-sum sequential games.
\newblock In \emph{International Conference on Artificial Intelligence and Statistics}, pages 2152--2160, 2024.

\bibitem[Brown and Sandholm(2019)]{brown2019solving}
N.~Brown and T.~Sandholm.
\newblock Solving imperfect-information games via discounted regret minimization.
\newblock In \emph{Proceedings of the AAAI Conference on Artificial Intelligence}, volume~33, pages 1829--1836, 2019.

\bibitem[Brown et~al.(2019)Brown, Lerer, Gross, and Sandholm]{brown2019deep}
N.~Brown, A.~Lerer, S.~Gross, and T.~Sandholm.
\newblock Deep counterfactual regret minimization.
\newblock In \emph{International conference on machine learning}, pages 793--802, 2019.

\bibitem[Cai et~al.(2025)Cai, Farina, Grand{-}Cl{\'{e}}ment, Kroer, Lee, Luo, and Zheng]{cai2023last}
Y.~Cai, G.~Farina, J.~Grand{-}Cl{\'{e}}ment, C.~Kroer, C.~Lee, H.~Luo, and W.~Zheng.
\newblock Last-iterate convergence properties of regret-matching algorithms in games.
\newblock In \emph{The Thirteenth International Conference on Learning Representations}, 2025.

\bibitem[Farina and Gatti(2017)]{farina2017extensive}
G.~Farina and N.~Gatti.
\newblock Extensive-form perfect equilibrium computation in two-player games.
\newblock In \emph{Proceedings of the AAAI Conference on Artificial Intelligence}, volume~31, 2017.

\bibitem[Farina et~al.(2017)Farina, Kroer, and Sandholm]{farina2017regret}
G.~Farina, C.~Kroer, and T.~Sandholm.
\newblock Regret minimization in behaviorally-constrained zero-sum games.
\newblock In \emph{International Conference on Machine Learning}, pages 1107--1116, 2017.

\bibitem[Farina et~al.(2018)Farina, Gatti, and Sandholm]{farina2018practical}
G.~Farina, N.~Gatti, and T.~Sandholm.
\newblock Practical exact algorithm for trembling-hand equilibrium refinements in games.
\newblock \emph{Advances in neural information processing systems}, 31, 2018.

\bibitem[Farina et~al.(2021)Farina, Kroer, and Sandholm]{farina2021faster}
G.~Farina, C.~Kroer, and T.~Sandholm.
\newblock Faster game solving via predictive blackwell approachability: Connecting regret matching and mirror descent.
\newblock In \emph{Proceedings of the AAAI Conference on Artificial Intelligence}, volume~35, pages 5363--5371, 2021.

\bibitem[Hart and Mas-Colell(2000)]{hart2000simple}
S.~Hart and A.~Mas-Colell.
\newblock A simple adaptive procedure leading to correlated equilibrium.
\newblock \emph{Econometrica}, 68\penalty0 (5):\penalty0 1127--1150, 2000.

\bibitem[Hoda et~al.(2010)Hoda, Gilpin, Pena, and Sandholm]{hoda2010smoothing}
S.~Hoda, A.~Gilpin, J.~Pena, and T.~Sandholm.
\newblock Smoothing techniques for computing nash equilibria of sequential games.
\newblock \emph{Mathematics of Operations Research}, 35\penalty0 (2):\penalty0 494--512, 2010.

\bibitem[Hsieh et~al.(2021)Hsieh, Antonakopoulos, and Mertikopoulos]{hsieh2021adaptive}
Y.-G. Hsieh, K.~Antonakopoulos, and P.~Mertikopoulos.
\newblock Adaptive learning in continuous games: Optimal regret bounds and convergence to nash equilibrium.
\newblock In \emph{Conference on Learning Theory}, pages 2388--2422, 2021.

\bibitem[Koller et~al.(1996)Koller, Megiddo, and Von~Stengel]{koller1996efficient}
D.~Koller, N.~Megiddo, and B.~Von~Stengel.
\newblock Efficient computation of equilibria for extensive two-person games.
\newblock \emph{Games and economic behavior}, 14\penalty0 (2):\penalty0 247--259, 1996.

\bibitem[Kroer et~al.(2017{\natexlab{a}})Kroer, Farina, and Sandholm]{kroer2017smoothing}
C.~Kroer, G.~Farina, and T.~Sandholm.
\newblock Smoothing method for approximate extensive-form perfect equilibrium.
\newblock In \emph{Proceedings of the Twenty-Sixth International Joint Conference on Artificial Intelligence}, pages 295--301, 2017{\natexlab{a}}.

\bibitem[Kroer et~al.(2017{\natexlab{b}})Kroer, Waugh, Kilin{\c{c}}{-}Karzan, and Sandholm]{kroer2017theoretical}
C.~Kroer, K.~Waugh, F.~Kilin{\c{c}}{-}Karzan, and T.~Sandholm.
\newblock Theoretical and practical advances on smoothing for extensive-form games.
\newblock In \emph{Proceedings of the 2017 {ACM} Conference on Economics and Computation, {EC} '17, Cambridge, MA, USA, June 26-30, 2017}, page 693, 2017{\natexlab{b}}.

\bibitem[Kuhn(1950)]{kuhn1950simplified}
H.~W. Kuhn.
\newblock A simplified two-person poker.
\newblock \emph{Contributions to the Theory of Games}, 1\penalty0 (97-103):\penalty0 2, 1950.

\bibitem[Lee et~al.(2021)Lee, Kroer, and Luo]{lee2021last}
C.-W. Lee, C.~Kroer, and H.~Luo.
\newblock Last-iterate convergence in extensive-form games.
\newblock \emph{Advances in Neural Information Processing Systems}, 34:\penalty0 14293--14305, 2021.

\bibitem[Lis{\`y} et~al.(2015)Lis{\`y}, Lanctot, and Bowling]{lisy2015online}
V.~Lis{\`y}, M.~Lanctot, and M.~H. Bowling.
\newblock Online monte carlo counterfactual regret minimization for search in imperfect information games.
\newblock In \emph{AAMAS}, pages 27--36, 2015.

\bibitem[Liu et~al.(2023)Liu, Ozdaglar, Yu, and Zhang]{liu2022power}
M.~Liu, A.~E. Ozdaglar, T.~Yu, and K.~Zhang.
\newblock The power of regularization in solving extensive-form games.
\newblock In \emph{The Eleventh International Conference on Learning Representations}, 2023.

\bibitem[Meng et~al.(2023)Meng, Ge, Li, An, and Gao]{meng2023efficient}
L.~Meng, Z.~Ge, W.~Li, B.~An, and Y.~Gao.
\newblock Efficient last-iterate convergence algorithms in solving games.
\newblock \emph{arXiv preprint arXiv:2308.11256}, 2023.

\bibitem[Mertikopoulos et~al.(2018)Mertikopoulos, Papadimitriou, and Piliouras]{mertikopoulos2018cycles}
P.~Mertikopoulos, C.~Papadimitriou, and G.~Piliouras.
\newblock Cycles in adversarial regularized learning.
\newblock In \emph{Proceedings of the twenty-ninth annual ACM-SIAM symposium on discrete algorithms}, pages 2703--2717, 2018.

\bibitem[Miltersen and S{\o}rensen(2010)]{miltersen2010computing}
P.~B. Miltersen and T.~B. S{\o}rensen.
\newblock Computing a quasi-perfect equilibrium of a two-player game.
\newblock \emph{Economic Theory}, 42:\penalty0 175--192, 2010.

\bibitem[Nesterov(2005)]{nesterov2005excessive}
Y.~Nesterov.
\newblock Excessive gap technique in nonsmooth convex minimization.
\newblock \emph{SIAM Journal on Optimization}, 16\penalty0 (1):\penalty0 235--249, 2005.

\bibitem[Perolat et~al.(2021)Perolat, Munos, Lespiau, Omidshafiei, Rowland, Ortega, Burch, Anthony, Balduzzi, De~Vylder, et~al.]{perolat2021poincare}
J.~Perolat, R.~Munos, J.-B. Lespiau, S.~Omidshafiei, M.~Rowland, P.~Ortega, N.~Burch, T.~Anthony, D.~Balduzzi, B.~De~Vylder, et~al.
\newblock From poincar{\'e} recurrence to convergence in imperfect information games: Finding equilibrium via regularization.
\newblock In \emph{International Conference on Machine Learning}, pages 8525--8535, 2021.

\bibitem[Perolat et~al.(2022)Perolat, De~Vylder, Hennes, Tarassov, Strub, de~Boer, Muller, Connor, Burch, Anthony, et~al.]{perolat2022mastering}
J.~Perolat, B.~De~Vylder, D.~Hennes, E.~Tarassov, F.~Strub, V.~de~Boer, P.~Muller, J.~T. Connor, N.~Burch, T.~Anthony, et~al.
\newblock Mastering the game of stratego with model-free multiagent reinforcement learning.
\newblock \emph{Science}, 378\penalty0 (6623):\penalty0 990--996, 2022.

\bibitem[Romanovskii(1962)]{romanovskii1962reduction}
I.~Romanovskii.
\newblock Reduction of a game with complete memory to a matrix game.
\newblock \emph{Soviet Mathematics}, 3:\penalty0 678--681, 1962.

\bibitem[Ross(1971)]{ross1971goofspiel}
S.~M. Ross.
\newblock Goofspiel—the game of pure strategy.
\newblock \emph{Journal of Applied Probability}, 8\penalty0 (3):\penalty0 621--625, 1971.

\bibitem[Selten(1965)]{selten1965spieltheoretische}
R.~Selten.
\newblock Spieltheoretische behandlung eines oligopolmodells mit nachfragetr{\"a}gheit: Teil i: Bestimmung des dynamischen preisgleichgewichts.
\newblock \emph{Zeitschrift f{\"u}r die gesamte Staatswissenschaft/Journal of Institutional and Theoretical Economics}, H. 2:\penalty0 301--324, 1965.

\bibitem[Selten(1975)]{selten1975reexamination}
R.~Selten.
\newblock Reexamination of the perfectness concept for equilibrium points in extensive games.
\newblock \emph{International Journal of Game Theory}, 4, 1975.

\bibitem[Southey et~al.(2005)Southey, Bowling, Larson, Piccione, Burch, Billings, and Rayner]{southey2012bayes}
F.~Southey, M.~H. Bowling, B.~Larson, C.~Piccione, N.~Burch, D.~Billings, and D.~C. Rayner.
\newblock Bayes? bluff: Opponent modelling in poker.
\newblock In \emph{{UAI} '05, Proceedings of the 21st Conference in Uncertainty in Artificial Intelligence, Edinburgh, Scotland, July 26-29, 2005}, pages 550--558, 2005.

\bibitem[Srinivasan et~al.(2018)Srinivasan, Lanctot, Zambaldi, P{\'e}rolat, Tuyls, Munos, and Bowling]{srinivasan2018actor}
S.~Srinivasan, M.~Lanctot, V.~Zambaldi, J.~P{\'e}rolat, K.~Tuyls, R.~Munos, and M.~Bowling.
\newblock Actor-critic policy optimization in partially observable multiagent environments.
\newblock \emph{Advances in neural information processing systems}, 31, 2018.

\bibitem[Tammelin(2014)]{tammelin2014solving}
O.~Tammelin.
\newblock Solving large imperfect information games using cfr+.
\newblock \emph{arXiv preprint arXiv:1407.5042}, 2014.

\bibitem[Van~Damme(1984)]{van1984relation}
E.~Van~Damme.
\newblock A relation between perfect equilibria in extensive form games and proper equilibria in normal form games.
\newblock \emph{International Journal of Game Theory}, 13:\penalty0 1--13, 1984.

\bibitem[Von~Stengel(1996)]{von1996efficient}
B.~Von~Stengel.
\newblock Efficient computation of behavior strategies.
\newblock \emph{Games and Economic Behavior}, 14\penalty0 (2):\penalty0 220--246, 1996.

\bibitem[Wei et~al.(2021)Wei, Lee, Zhang, and Luo]{wei2020linear}
C.~Wei, C.~Lee, M.~Zhang, and H.~Luo.
\newblock Linear last-iterate convergence in constrained saddle-point optimization.
\newblock In \emph{9th International Conference on Learning Representations}, 2021.

\bibitem[Zinkevich et~al.(2007)Zinkevich, Johanson, Bowling, and Piccione]{zinkevich2007regret}
M.~Zinkevich, M.~Johanson, M.~Bowling, and C.~Piccione.
\newblock Regret minimization in games with incomplete information.
\newblock \emph{Advances in neural information processing systems}, 20, 2007.

\end{thebibliography}
\newpage
\onecolumn
\appendix
\section{Proof of Theorem~\ref{thm:best_iterate_converge_of_RTRM}}\label{sec:proof_RTRM_convergence}

\begin{proof}
The proof proceeds in two steps: (1) establishing best-iterate convergence in the coordinate space of the transformed game \(\hat{\Gamma}\) using Online Mirror Descent (OMD), and (2) mapping the result to the original strategy space of the \(\epsilon\)-perturbed game \(\Gamma^\epsilon\), preserving perturbation constraints.

\textbf{Step 1: Convergence in the Transformed Game \(\hat{\Gamma}\).} 
Following \citet{meng2023efficient}, we reformulate Reward Transformation Regret Minimization (RTRM) dynamics as an OMD process, leveraging the equivalence between Regret Matching (RM) and OMD \cite{farina2021faster}. The RTRM dynamics are defined as:
\begin{align}
    \boldsymbol{v}_i^t &= \boldsymbol{U}_i \boldsymbol{x}_{-i}^t + \mu (\boldsymbol{x}_i^{r} - \boldsymbol{x}_i^t), \label{eq:RTRM_v} \\
    \boldsymbol{r}_i^t &= \boldsymbol{v}_i^t - \langle \boldsymbol{v}_i^t, \boldsymbol{x}_i^t \rangle \boldsymbol{1}, \label{eq:RTRM_r} \\
    \boldsymbol{R}_i^t &= \begin{cases}
        \boldsymbol{R}_i^{t-1} + \boldsymbol{r}_i^t, & \text{if RM}, \\
        [\boldsymbol{R}_i^{t-1} + \boldsymbol{r}_i^t]^+, & \text{if RM+}, \\
        \frac{t^\alpha}{t^\alpha + 1} [\boldsymbol{R}_i^{t-1} + \boldsymbol{r}_i^t]^+ + \frac{t^\beta}{t^\beta + 1} [\boldsymbol{R}_i^{t-1} + \boldsymbol{r}_i^t]^-, & \text{if DRM},
    \end{cases} \label{eq:RTRM_R} \\
    \boldsymbol{x}_i^{t+1} &= \begin{cases}
        \frac{[\boldsymbol{R}_i^t]^+}{\|[\boldsymbol{R}_i^t]^+\|_1}, & \text{if } \|[\boldsymbol{R}_i^t]^+\|_1 > 0, \\
        \frac{\boldsymbol{1}}{|A_i|}, & \text{otherwise},
    \end{cases} \label{eq:RTRM_x}
\end{align}
where \([\cdot]^- = \min(\cdot, \boldsymbol{0})\), and \(\alpha, \beta \in \mathbb{R}\) are discount parameters for positive and negative cumulative regrets, respectively.

In Online Mirror Descent (OMD), given a strategy \(\boldsymbol{x}_i^t \in \mathcal{X}_i\) at iteration \(t\) and an observed loss vector \(\boldsymbol{\ell}_i^t\), the strategy is updated as:
\[
    \boldsymbol{x}_i^{t+1} = \argmin_{\boldsymbol{x}_i \in \mathcal{X}_i} \left\{ \langle \boldsymbol{\ell}_i^t, \boldsymbol{x}_i \rangle + \frac{1}{\eta} D_{\psi}(\boldsymbol{x}_i, \boldsymbol{x}_i^t) \right\},
\]
where \(\eta > 0\) is the step size, and \(D_{\psi}(\boldsymbol{x}, \boldsymbol{y}) = \psi(\boldsymbol{x}) - \psi(\boldsymbol{y}) - \langle \nabla \psi(\boldsymbol{y}), \boldsymbol{x} - \boldsymbol{y} \rangle\) is the Bregman divergence with a distance-generating function \(\psi\). When \(\psi(\boldsymbol{x}) = \frac{1}{2} \|\boldsymbol{x}\|_2^2\), OMD reduces to Gradient Descent Ascent (GDA).

We rely on two key lemmas to establish convergence:

\begin{lemma}[Saddle-Point Metric Subregularity \cite{wei2020linear}]\label{lemma:MS}
    Let \(\boldsymbol{x}^* \in \mathcal{X}^*\) be the Nash equilibrium of the game. There exists a constant \(C > 0\), dependent only on the game, such that for any strategy \(\boldsymbol{x} \in \mathcal{X} \setminus \mathcal{X}^*\):
    \[
        \|\boldsymbol{x}^* - \boldsymbol{x}\|_2 \leq \frac{\text{Exp}(\boldsymbol{x})}{C}.
    \]
\end{lemma}

\begin{lemma}[OMD Equivalence for RTRM \cite{farina2021faster}]\label{lem:omd_equivalence}
    RTRM updates are equivalent to GDA updates:
    \[
        \boldsymbol{\theta}_i^{t+1,n} \in \argmin_{\boldsymbol{\theta}_i \in \mathbb{R}^{|A_i|}} \left\{ \langle -\boldsymbol{r}_i(\boldsymbol{\theta}_i^{t,n}), \boldsymbol{\theta}_i \rangle + \frac{1}{\eta} D_{\psi}(\boldsymbol{\theta}_i, \boldsymbol{\theta}_i^{t,n}) \right\},
    \]
    where \(\eta > 0\), \(\boldsymbol{r}_i(\boldsymbol{\theta}_i^{t,n}) = \langle \boldsymbol{\ell}_i^{t,n}, \boldsymbol{x}_i^{t,n} \rangle \boldsymbol{1} - \boldsymbol{\ell}_i^{t,n}\), \(\boldsymbol{\ell}_i^{t,n} = \boldsymbol{g}_i^{t,n} + \mu (\boldsymbol{x}_i^{t,n} - \boldsymbol{x}_i^{r,n})\), \(\boldsymbol{g}_i^{t,n} = -\boldsymbol{U}_i \boldsymbol{x}_{-i}^{t,n}\), and \(\psi(\boldsymbol{\theta}) = \frac{1}{2} \|\boldsymbol{\theta}\|_2^2\). The closed-form solution is:
    \[
        \boldsymbol{\theta}_i^{t+1,n} = \begin{cases}
            \boldsymbol{\theta}_i^{t,n} + \eta \boldsymbol{r}_i(\boldsymbol{\theta}_i^{t,n}), & \text{for RM}, \\
            [\boldsymbol{\theta}_i^{t,n} + \eta \boldsymbol{r}_i(\boldsymbol{\theta}_i^{t,n})]^+, & \text{for RM+}, \\
            \frac{t^\alpha}{t^\alpha + 1} [\boldsymbol{\theta}_i^{t,n} + \eta \boldsymbol{r}_i(\boldsymbol{\theta}_i^{t,n})]^+ + \frac{t^\beta}{t^\beta + 1} [\boldsymbol{\theta}_i^{t,n} + \eta \boldsymbol{r}_i(\boldsymbol{\theta}_i^{t,n})]^-, & \text{for DRM}.
        \end{cases}
    \]
\end{lemma}

Using Lemma~\ref{lem:omd_equivalence}, we establish that RTRM corresponds to GDA, a special case of OMD with a Euclidean regularizer \cite{meng2023efficient, farina2021faster}. By convex analysis (Appendix E of \citet{meng2023efficient}), we bound the Bregman divergence between the saddle point \(\hat{\boldsymbol{x}}^{*,n} = \frac{\boldsymbol{\theta}^{*,n}}{\|\boldsymbol{\theta}^{*,n}\|_1}\) and the strategy \(\hat{\boldsymbol{x}}^{t,n} = \frac{\boldsymbol{\theta}^{t,n}}{\|\boldsymbol{\theta}^{t,n}\|_1}\):
\[
    \sum_{t=1}^T C_1 D_{\psi}(\hat{\boldsymbol{x}}^{*,n}, \hat{\boldsymbol{x}}^{t,n}) \leq C_2,
\]
where \(C_1 = 2\eta \mu - (\eta C_0)^2\), \(C_2 = D_{\psi}(\boldsymbol{\theta}^{1,n,*}, \boldsymbol{\theta}^{1,n}) + \eta \langle -\boldsymbol{r}_i(\boldsymbol{\theta}^{*,n}), \boldsymbol{\theta}^{1,n} \rangle\), and \(C_0 = 2P^2 + 3\mu P + P + \mu\), with \(P = \max(|A_1|, |A_2|)\). Here, \(\boldsymbol{\theta}^{1,n,*}\) is the projection of \(\boldsymbol{\theta}^{1,n}\) onto the saddle-point ray.

Applying the Mean Value Theorem, there exists \(t \in \{1, \dots, T\}\) such that:
\begin{equation}
    \|\hat{\boldsymbol{x}}^{*,n} - \hat{\boldsymbol{x}}^{t,n}\|_2 \leq \sqrt{\frac{2C_2}{C_1 T}},
    \label{eq:best_rtrm}
\end{equation}
establishing a convergence rate of \(O(1/\sqrt{T})\).

\textbf{Step 2: Mapping to the Perturbed Game \(\Gamma^\epsilon\).}
We map the convergence result to the original strategy space of \(\Gamma^\epsilon\):
\[
    \|\boldsymbol{x}^{*,n} - \boldsymbol{x}^{t,n}\|_2 = \|\boldsymbol{B} \hat{\boldsymbol{x}}^{*,n} - \boldsymbol{B} \hat{\boldsymbol{x}}^{t,n}\|_2 \leq \|\boldsymbol{B}\|_2 \|\hat{\boldsymbol{x}}^{*,n} - \hat{\boldsymbol{x}}^{t,n}\|_2 \leq \sqrt{\frac{2C_2}{C_1 T}},
\]
where \(\|\boldsymbol{B}\|_2 = \sqrt{\lambda_{\max}(\boldsymbol{B}^\top \boldsymbol{B})} \leq 1\) for \(\epsilon \geq 0\), as the basis matrix \(\boldsymbol{B}\) is constructed to preserve perturbation constraints. Thus, RTRM achieves best-iterate convergence to the saddle point of the \(\epsilon\)-perturbed RT-BSPP at a rate of \(O(1/\sqrt{T})\).
\end{proof}

\section{Proof of Theorem~\ref{thm:asymptotic_last_iterate_converge_of_RTRM}}\label{sec:proof_asymptotic_RTRM}
\begin{proof}
The proof proceeds in two steps: (1) establishing that the saddle point \(\boldsymbol{x}^{*,n}\) of the \(n\)-th RT-BSPP is closer to the Nash Equilibrium (NE) \(\boldsymbol{x}^*\) of \(\Gamma^\epsilon\) than the reference strategy \(\boldsymbol{x}^{r,n}\), and (2) proving that the sequence \(\{\boldsymbol{x}^{*,n}\}\) converges to \(\boldsymbol{x}^*\), and combining with Theorem~\ref{thm:best_iterate_converge_of_RTRM}, showing that the strategy \(\boldsymbol{x}^{t,n}\) asymptotically converges to \(\boldsymbol{x}^*\).

\begin{lemma}\label{le:relation_of_3_points}
Let \(\boldsymbol{x}^{r,n}\) be the reference strategy, \(\mu > 0\) the RT weight, \(\boldsymbol{x}^{*,n} \in \mathcal{X}^{*,n}\) the saddle point of the \(n\)-th RT-BSPP, and \(\boldsymbol{x}^* \in \mathcal{X}^*\) the NE of \(\Gamma^\epsilon\). Provided \(\boldsymbol{x}^{r,n} \neq \boldsymbol{x}^{*,n} \neq \boldsymbol{x}^*\), the following hold:
\begin{equation}
    \|\boldsymbol{x}^* - \boldsymbol{x}^{*,n}\|_2 \leq \|\boldsymbol{x}^* - \boldsymbol{x}^{r,n}\|_2,
    \label{eq:saddle_reference_NE}
\end{equation}
and
\begin{equation}
    \|\boldsymbol{x}^* - \boldsymbol{x}^{r,n}\|_2 - \|\boldsymbol{x}^* - \boldsymbol{x}^{*,n}\|_2 \geq \frac{C^2}{\mu^2 (\|\boldsymbol{x}^* - \boldsymbol{x}^{r,n}\|_2 + \|\boldsymbol{x}^* - \boldsymbol{x}^{*,n}\|_2)},
    \label{eq:bound_reference_saddle}
\end{equation}
where \(C > 0\) is a constant depending only on the game.
\end{lemma}

\begin{proof}
Let \(\boldsymbol{x}^{*,n}\) be the saddle point of the \(n\)-th RT-BSPP, and \(\boldsymbol{x}^*\) the NE of the \(\epsilon\)-perturbed game \(\Gamma^\epsilon\). By the saddle-point property, for player \(i \in \{1, 2\}\):
\[
    \langle \boldsymbol{x}_i^*, \boldsymbol{\ell}_i^{*,n} \rangle \geq \langle \boldsymbol{x}_i^{*,n}, \boldsymbol{\ell}_i^{*,n} \rangle,
\]
where \(\boldsymbol{\ell}_i^{*,n} = \boldsymbol{g}_i^{*,n} + \mu (\boldsymbol{x}_i^{*,n} - \boldsymbol{x}_i^{r,n})\), and \(\boldsymbol{g}_i^{*,n} = -\boldsymbol{U}_i \boldsymbol{x}_{-i}^{*,n}\) is the loss gradient for player \(i\) given the opponent’s strategy \(\boldsymbol{x}_{-i}^{*,n}\).

Summing over both players:
\[
    \sum_{i \in \{1, 2\}} \langle \boldsymbol{x}_i^*, \boldsymbol{g}_i^{*,n} + \mu (\boldsymbol{x}_i^{*,n} - \boldsymbol{x}_i^{r,n}) \rangle \geq \sum_{i \in \{1, 2\}} \langle \boldsymbol{x}_i^{*,n}, \boldsymbol{g}_i^{*,n} + \mu (\boldsymbol{x}_i^{*,n} - \boldsymbol{x}_i^{r,n}) \rangle.
\]
Rearranging terms:
\begin{align}
    \mu \langle \boldsymbol{x}^*, \boldsymbol{x}^{*,n} - \boldsymbol{x}^{r,n} \rangle &\geq \mu \langle \boldsymbol{x}^{*,n}, \boldsymbol{x}^{*,n} - \boldsymbol{x}^{r,n} \rangle + \sum_{i \in \{1, 2\}} \left( \langle \boldsymbol{x}_i^{*,n}, \boldsymbol{g}_i^{*,n} \rangle - \langle \boldsymbol{x}_i^*, \boldsymbol{g}_i^{*,n} \rangle \right). \label{eq:convert_to_exp}
\end{align}
Since \(\sum_{i \in \{1, 2\}} \left( \langle \boldsymbol{x}_i^{*,n}, \boldsymbol{g}_i^{*,n} \rangle - \langle \boldsymbol{x}_i^*, \boldsymbol{g}_i^{*,n} \rangle \right) = \text{Exp}(\boldsymbol{x}^{*,n}) \geq 0\), we have:
\[
    \mu \langle \boldsymbol{x}^*, \boldsymbol{x}^{*,n} - \boldsymbol{x}^{r,n} \rangle \geq \mu \langle \boldsymbol{x}^{*,n}, \boldsymbol{x}^{*,n} - \boldsymbol{x}^{r,n} \rangle.
\]
This simplifies to:
\begin{align}
    \langle \boldsymbol{x}^{*,n} - \boldsymbol{x}^{r,n}, \boldsymbol{x}^* - \boldsymbol{x}^{*,n} \rangle &\geq \langle \boldsymbol{x}^{*,n} - \boldsymbol{x}^{r,n}, \boldsymbol{x}^* - \boldsymbol{x}^{r,n} \rangle. \label{eq:q1}
\end{align}
Adding \(\langle \boldsymbol{x}^{r,n} - \boldsymbol{x}^*, \boldsymbol{x}^* - \boldsymbol{x}^{*,n} \rangle\) to both sides:
\[
    \langle \boldsymbol{x}^{*,n} - \boldsymbol{x}^{r,n}, \boldsymbol{x}^* - \boldsymbol{x}^{*,n} \rangle + \langle \boldsymbol{x}^{r,n} - \boldsymbol{x}^*, \boldsymbol{x}^* - \boldsymbol{x}^{*,n} \rangle \geq \langle \boldsymbol{x}^{*,n} - \boldsymbol{x}^{r,n}, \boldsymbol{x}^* - \boldsymbol{x}^{r,n} \rangle + \langle \boldsymbol{x}^{r,n} - \boldsymbol{x}^*, \boldsymbol{x}^* - \boldsymbol{x}^{*,n} \rangle.
\]
This yields:
\[
    \langle \boldsymbol{x}^{*,n} - \boldsymbol{x}^*, \boldsymbol{x}^* - \boldsymbol{x}^{*,n} \rangle \geq \langle \boldsymbol{x}^{r,n} - \boldsymbol{x}^*, \boldsymbol{x}^* - \boldsymbol{x}^{r,n} \rangle,
\]
or equivalently:
\[
    \|\boldsymbol{x}^* - \boldsymbol{x}^{*,n}\|_2^2 \leq \|\boldsymbol{x}^* - \boldsymbol{x}^{r,n}\|_2^2 - \|\boldsymbol{x}^{r,n} - \boldsymbol{x}^{*,n}\|_2^2,
\]
proving Equation~\eqref{eq:saddle_reference_NE}. Assuming \(\boldsymbol{x}^{r,n} \neq \boldsymbol{x}^{*,n} \neq \boldsymbol{x}^*\), we apply the identity \(a^2 - b^2 = (a - b)(a + b)\) to obtain:
\[
    \|\boldsymbol{x}^* - \boldsymbol{x}^{r,n}\|_2^2 - \|\boldsymbol{x}^* - \boldsymbol{x}^{*,n}\|_2^2 = \left( \|\boldsymbol{x}^* - \boldsymbol{x}^{r,n}\|_2 - \|\boldsymbol{x}^* - \boldsymbol{x}^{*,n}\|_2 \right) \left( \|\boldsymbol{x}^* - \boldsymbol{x}^{r,n}\|_2 + \|\boldsymbol{x}^* - \boldsymbol{x}^{*,n}\|_2 \right).
\]
Thus:
\[
    \|\boldsymbol{x}^* - \boldsymbol{x}^{r,n}\|_2 - \|\boldsymbol{x}^* - \boldsymbol{x}^{*,n}\|_2 \geq \frac{\|\boldsymbol{x}^{r,n} - \boldsymbol{x}^{*,n}\|_2^2}{\|\boldsymbol{x}^* - \boldsymbol{x}^{r,n}\|_2 + \|\boldsymbol{x}^* - \boldsymbol{x}^{*,n}\|_2}. \label{eq:relation_2}
\]

To bound \(\|\boldsymbol{x}^{r,n} - \boldsymbol{x}^{*,n}\|_2\), return to Equation~\eqref{eq:convert_to_exp}:
\[
    \text{Exp}(\boldsymbol{x}^{*,n}) \leq \mu \langle \boldsymbol{x}^* - \boldsymbol{x}^{*,n}, \boldsymbol{x}^{*,n} - \boldsymbol{x}^{r,n} \rangle.
\]
By the Cauchy-Schwarz inequality:
\[
    \text{Exp}(\boldsymbol{x}^{*,n}) \leq \mu \|\boldsymbol{x}^* - \boldsymbol{x}^{*,n}\|_2 \|\boldsymbol{x}^{*,n} - \boldsymbol{x}^{r,n}\|_2. \label{eq:exp_saddle}
\]
By Lemma~\ref{lemma:MS}, there exists a constant \(C > 0\) such that:
\[
    C \|\boldsymbol{x}^* - \boldsymbol{x}^{*,n}\|_2 \leq \text{Exp}(\boldsymbol{x}^{*,n}) \leq \mu \|\boldsymbol{x}^* - \boldsymbol{x}^{*,n}\|_2 \|\boldsymbol{x}^{*,n} - \boldsymbol{x}^{r,n}\|_2.
\]
Assuming \(\boldsymbol{x}^* \neq \boldsymbol{x}^{*,n}\), we derive:
\[
    \|\boldsymbol{x}^{*,n} - \boldsymbol{x}^{r,n}\|_2 \geq \frac{C}{\mu}.
\]
Substituting into Equation~\eqref{eq:relation_2}:
\[
    \|\boldsymbol{x}^* - \boldsymbol{x}^{r,n}\|_2 - \|\boldsymbol{x}^* - \boldsymbol{x}^{*,n}\|_2 \geq \frac{C^2}{\mu^2 (\|\boldsymbol{x}^* - \boldsymbol{x}^{r,n}\|_2 + \|\boldsymbol{x}^* - \boldsymbol{x}^{*,n}\|_2)},
\]
proving Equation~\eqref{eq:bound_reference_saddle}.
\end{proof}

We now establish the asymptotic convergence of RTRM. Let \(\boldsymbol{x}^{t,n}\) denote the strategy at iteration \(t\) of the \(n\)-th RT-BSPP, \(\boldsymbol{x}^{*,n}\) its saddle point, and \(\boldsymbol{x}^*\) the NE of \(\Gamma^\epsilon\). By the triangle inequality:
\[
    \|\boldsymbol{x}^* - \boldsymbol{x}^{t,n}\|_2 \leq \|\boldsymbol{x}^* - \boldsymbol{x}^{*,n}\|_2 + \|\boldsymbol{x}^{*,n} - \boldsymbol{x}^{t,n}\|_2. \label{eq:bound}
\]
Here, \(\|\boldsymbol{x}^* - \boldsymbol{x}^{*,n}\|_2\) measures the distance between the NE and the \(n\)-th saddle point, and \(\|\boldsymbol{x}^{*,n} - \boldsymbol{x}^{t,n}\|_2\) captures the error within the \(n\)-th RT-BSPP.

By the RT framework initialization (Equation~\eqref{eq:rt_init}), \(\boldsymbol{x}^{r,n} = \boldsymbol{x}^{T+1,n-1} = \boldsymbol{x}^{*,n-1}\). Substituting into Equation~\eqref{eq:bound_reference_saddle}:
\[
    \|\boldsymbol{x}^* - \boldsymbol{x}^{*,n-1}\|_2 - \|\boldsymbol{x}^* - \boldsymbol{x}^{*,n}\|_2 \geq \frac{C^2}{\mu^2 (\|\boldsymbol{x}^* - \boldsymbol{x}^{*,n-1}\|_2 + \|\boldsymbol{x}^* - \boldsymbol{x}^{*,n}\|_2)}. \label{eq:difference_reference}
\]
This implies that the sequence \(\{\|\boldsymbol{x}^* - \boldsymbol{x}^{*,n}\|_2\}\) is strictly decreasing and non-negative. Since the right-hand side of Equation~\eqref{eq:difference_reference} is positive and increases as \(\boldsymbol{x}^{*,n-1}\) approaches \(\boldsymbol{x}^*\), the sequence converges to a limit \(L = 0\) by the monotone convergence theorem. Meanwhile, Theorem~\ref{thm:best_iterate_converge_of_RTRM} guarantees that \(\|\boldsymbol{x}^{*,n} - \boldsymbol{x}^{t,n}\|_2 \to 0\) as \(t \to \infty\). Thus:
\[
    \lim_{n \to \infty, t \to \infty} \|\boldsymbol{x}^* - \boldsymbol{x}^{t,n}\|_2 = 0,
\]
establishing the asymptotic last-iterate convergence of RTRM to the NE \(\boldsymbol{x}^*\) of \(\Gamma^\epsilon\).
\end{proof}

\section{Proof of Theorem~\ref{thm:best_iterate_convergence_of_RTCFR}}\label{sec:proof_RTCFR_convergence}

\begin{proof}
In an extensive-form game (EFG), Reward Transformation Counterfactual Regret Minimization (RTCFR) optimizes sequence-form strategies using a bottom-up approach, analogous to Counterfactual Regret Minimization (CFR). We establish best-iterate convergence of the sequence-form strategy \(\boldsymbol{q}^{t,n}\) to the saddle point \(\boldsymbol{q}^{*,n}\) of the \(\epsilon\)-perturbed \(n\)-th RT-BSPP.

For each information set \(I \in \mathcal{I}_i\) of player \(i \in \{1, 2\}\), RTCFR updates the local behavioral strategy \(\boldsymbol{x}^{t,n}(I)\) using dynamics similar to RTRM (Equations~\eqref{eq:RTRM_v}--\eqref{eq:RTRM_x}). By Theorem~\ref{thm:best_iterate_converge_of_RTRM}, the local convergence at each information set is bounded as:
\begin{equation}
    \|\boldsymbol{x}^{*,n}(I) - \boldsymbol{x}^{t,n}(I)\|_2 \leq \sqrt{\frac{2 C_2^{\max}}{C_1^{\min} T}},
    \label{eq:bound_local}
\end{equation}
where \(C_1^{\min} = \min_{I \in \mathcal{I}_1 \cup \mathcal{I}_2} C_1^I\), \(C_2^{\max} = \max_{I \in \mathcal{I}_1 \cup \mathcal{I}_2} C_2^I\), and \(C_1^I, C_2^I\) are constants defined in Theorem~\ref{thm:best_iterate_converge_of_RTRM} for information set \(I\).

Let \(\mathcal{Q} \subseteq \mathbb{R}_{\geq 0}^{|\mathcal{X}|}\) denote the sequence-form strategy space, satisfying \(\boldsymbol{q}[\emptyset] = 1\) and \(\sum_{a \in A(I)} \boldsymbol{q}[Ia] = \boldsymbol{q}[pI]\) for all \(I \in \mathcal{I}_i\), where \(pI\) is the parent sequence of \(I\). Define \(M_{\mathcal{Q}} = \max_{\boldsymbol{q} \in \mathcal{Q}} \|\boldsymbol{q}\|_1\), representing the maximum number of information sets with non-zero reach probability under a pure strategy. We define the distance between the saddle-point sequence-form strategy \(\boldsymbol{q}^{*,n}\) and the iterate \(\boldsymbol{q}^{t,n}\) as:
\begin{equation}
    \|\boldsymbol{q}^{*,n} - \boldsymbol{q}^{t,n}\|_2 = \sqrt{\sum_{I \in \mathcal{I}_1 \cup \mathcal{I}_2} \boldsymbol{q}^{*,n}[pI]^2 \|\boldsymbol{x}^{*,n}(I) - \boldsymbol{x}^{t,n}(I)\|_2^2},
    \label{eq:sequence_norm}
\end{equation}
where \(\boldsymbol{x}^{*,n}(I), \boldsymbol{x}^{t,n}(I) \in \Delta^{|A(I)|}\) are the behavioral strategies at information set \(I\), and \(\boldsymbol{q}^{*,n}[pI]\) is the reach probability of \(I\).

Using Equation~\eqref{eq:bound_local}, we bound the squared distance and summing over all information sets:
\[
    \sum_{I \in \mathcal{I}_1 \cup \mathcal{I}_2} \boldsymbol{q}^{*,n}[pI]^2 \|\boldsymbol{x}^{*,n}(I) - \boldsymbol{x}^{t,n}(I)\|_2^2 \leq \sum_{I \in \mathcal{I}_1 \cup \mathcal{I}_2} \boldsymbol{q}^{*,n}[pI]^2 \cdot \frac{2 C_2^{\max}}{C_1^{\min} T}.
\]
Since \(\sum_{I \in \mathcal{I}_1 \cup \mathcal{I}_2} \boldsymbol{q}^{*,n}[pI]^2 \leq M_{\mathcal{Q}}\) (as reach probabilities are non-negative and less than 1, bounded by the structure of the game tree), we have:
\[
    \|\boldsymbol{q}^{*,n} - \boldsymbol{q}^{t,n}\|_2^2 = \sum_{I \in \mathcal{I}_1 \cup \mathcal{I}_2} \boldsymbol{q}^{*,n}[pI]^2 \|\boldsymbol{x}^{*,n}(I) - \boldsymbol{x}^{t,n}(I)\|_2^2 \leq M_{\mathcal{Q}} \cdot \frac{2 C_2^{\max}}{C_1^{\min} T}.
\]
Taking the square root:
\[
    \|\boldsymbol{q}^{*,n} - \boldsymbol{q}^{t,n}\|_2 \leq \sqrt{M_{\mathcal{Q}} \cdot \frac{2 C_2^{\max}}{C_1^{\min} T}}.
\]
establishing the \(O(1/\sqrt{T})\) best-iterate convergence rate to the saddle point \(\boldsymbol{q}^{*,n}\) in the \(\epsilon\)-perturbed \(n\)-th RT-BSPP.
\end{proof}

\section{Proof of Theorem~\ref{thm:asymptotic_last_iterate_converge_of_RTCFR}}\label{sec:proof_asymptotic_RTCFR}

\begin{proof}
Let \(\boldsymbol{q}^* \in \mathcal{Q}^*\) be the NE of the \(\epsilon\)-perturbed game \(\Gamma^\epsilon\), and let \(\boldsymbol{q}^{*,n} \in \mathcal{Q}^{*,n}\) and \(\boldsymbol{q}^{r,n}\) be the saddle point and reference strategy of the \(n\)-th Reward Transformation Bilinear Saddle-Point Problem (RT-BSPP), respectively. We prove that the sequence \(\{\boldsymbol{q}^{*,n}\}\) converges to \(\boldsymbol{q}^*\), and combined with Theorem~\ref{thm:best_iterate_convergence_of_RTCFR}, the RTCFR strategy \(\boldsymbol{q}^{t,n}\) asymptotically converges to \(\boldsymbol{q}^*\).

For each information set \(I \in \mathcal{I}_i\) of player \(i \in \{1, 2\}\), let \(\triangle_I\) denotes the subtree rooted at \(I\), the saddle-point property of the \(n\)-th RT-BSPP implies:
\begin{equation}
    \boldsymbol{q}_i^{*,n}[pI] \langle \boldsymbol{x}_i^*(I), \boldsymbol{\ell}_i^{*,(*,n)}(I) \rangle \geq \boldsymbol{q}_i^{*,n}[pI] \langle \boldsymbol{x}_i^{*,n}(I), \boldsymbol{\ell}_i^{*,n}(I) \rangle,
    \label{eq:laminar_saddle_property}
\end{equation}
where \(\boldsymbol{x}_i^*(I), \boldsymbol{x}_i^{*,n}(I)\) are the behavioral strategies at \(I\), \(\boldsymbol{q}_i^{*,n}[pI]\) is the reach probability of \(I\), and the losses are:
\[
    \boldsymbol{\ell}_i^{*,(*,n)}(I) = \boldsymbol{g}_i^{*,n}(I) + \mu \left( \boldsymbol{x}_i^{*,n}(I) - \boldsymbol{x}_i^{r,n}(I) \right) + \left( \sum_{I' \in C(I,a)} V_{\triangle_{I'}}^{*,n}(\boldsymbol{x}_i^*(\triangle_{I'})) \right)_{a \in A(I)},
\]
\[
    \boldsymbol{\ell}_i^{*,n}(I) = \boldsymbol{g}_i^{*,n}(I) + \mu \left( \boldsymbol{x}_i^{*,n}(I) - \boldsymbol{x}_i^{r,n}(I) \right) + \left( \sum_{I' \in C(I,a)} V_{\triangle_{I'}}^{*,n}(\boldsymbol{x}_i^{*,n}(\triangle_{I'})) \right)_{a \in A(I)},
\]
with \(\boldsymbol{g}_i^{*,n}(I) = -\boldsymbol{U}_i(I) \boldsymbol{q}_{-i}^{*,n}\) as the loss gradient, and the subtree value defined recursively as:
\[
    V_{\triangle_I}^{*,n}(\boldsymbol{x}_i) = \langle \boldsymbol{x}_i(I), \boldsymbol{g}_i^{*,n}(I) \rangle + \sum_{a \in A(I)} \boldsymbol{x}_i[Ia] \sum_{I' \in C(I,a)} V_{\triangle_{I'}}^{*,n}(\boldsymbol{x}_i(\triangle_{I'})).
\]

Extending Equation~\eqref{eq:laminar_saddle_property}, we obtain:
    \begin{align}
        \langle \boldsymbol{q}_i^*(\triangle_I), -\boldsymbol{U}_i(\triangle_I) \boldsymbol{q}_{-i}^{*,n} \rangle + \langle \boldsymbol{x}_i^*(I), \mu (\boldsymbol{x}_i^{*,n}(I) - \boldsymbol{x}_i^{r,n}(I)) \rangle \geq \langle \boldsymbol{q}_i^{*,n}(\triangle_I), -\boldsymbol{U}_i(\triangle_I) \boldsymbol{q}_{-i}^{*,n} \rangle 
        + \boldsymbol{q}_i^{*,n}[pI] \langle \boldsymbol{x}_i^{*,n}(I), \mu (\boldsymbol{x}_i^{*,n}(I) - \boldsymbol{x}_i^{r,n}(I)) \rangle.
        \label{eq:extended_saddle}
    \end{align}
    Summing over subtrees rooted at information sets \(I\) for players \(i \in \{1, 2\}\), we have:
    \[
        \sum_{i \in \{1, 2\}} \langle \boldsymbol{q}_i^*(\triangle_I), -\boldsymbol{U}_i(\triangle_I) \boldsymbol{q}_{-i}^{*,n} \rangle \leq \sum_{i \in \{1, 2\}} \langle \boldsymbol{q}_i^{*,n}(\triangle_I), -\boldsymbol{U}_i(\triangle_I) \boldsymbol{q}_{-i}^{*,n} \rangle,
    \]
    yielding, for each information set:
    \[
        \boldsymbol{q}_i^{*,n}[pI] \langle \boldsymbol{x}_i^*(I), \mu (\boldsymbol{x}_i^{*,n}(I) - \boldsymbol{x}_i^{r,n}(I)) \rangle \geq \boldsymbol{q}_i^{*,n}[pI] \langle \boldsymbol{x}_i^{*,n}(I), \mu (\boldsymbol{x}_i^{*,n}(I) - \boldsymbol{x}_i^{r,n}(I)) \rangle.
    \]
Following the proof of Lemma~\ref{le:relation_of_3_points}, for each information set \(I \in \mathcal{I}_i\), we derive:
\[
    \boldsymbol{q}_i^{*,n}[pI] \left( \|\boldsymbol{x}_i^*(I) - \boldsymbol{x}_i^{r,n}(I)\|_2^2 - \|\boldsymbol{x}_i^*(I) - \boldsymbol{x}_i^{*,n}(I)\|_2^2 \right) \geq 0,
\]
implying:
\[
    \|\boldsymbol{x}_i^*(I) - \boldsymbol{x}_i^{*,n}(I)\|_2 \leq \|\boldsymbol{x}_i^*(I) - \boldsymbol{x}_i^{r,n}(I)\|_2.
\]
Applying Lemma~\ref{le:relation_of_3_points}, assuming \(\boldsymbol{x}_i^{r,n}(I) \neq \boldsymbol{x}_i^{*,n}(I) \neq \boldsymbol{x}_i^*(I)\), we obtain:
\[
    \boldsymbol{q}_i^{*,n}[pI] \left( \|\boldsymbol{x}_i^*(I) - \boldsymbol{x}_i^{r,n}(I)\|_2 - \|\boldsymbol{x}_i^*(I) - \boldsymbol{x}_i^{*,n}(I)\|_2 \right) \geq \boldsymbol{q}_i^{*,n}[pI] \frac{C_I^2}{\mu^2 \left( \|\boldsymbol{x}_i^*(I) - \boldsymbol{x}_i^{r,n}(I)\|_2 + \|\boldsymbol{x}_i^*(I) - \boldsymbol{x}_i^{*,n}(I)\|_2 \right)},
\]
where \(C_I > 0\) is a constant specific to information set \(I\). Define \(C_{\min} = \min_{I \in \mathcal{I}_1 \cup \mathcal{I}_2} \frac{C_I^2}{\mu^2 \left( \|\boldsymbol{x}_i^*(I) - \boldsymbol{x}_i^{r,n}(I)\|_2 + \|\boldsymbol{x}_i^*(I) - \boldsymbol{x}_i^{*,n}(I)\|_2 \right)}\).

For the sequence-form strategy, define the \(\ell_2\)-norm distance as:
\[
    \|\boldsymbol{q}^* - \boldsymbol{q}^{*,n}\|_2 = \sqrt{\sum_{i \in \{1, 2\}} \sum_{I \in \mathcal{I}_i} \boldsymbol{q}_i^{*,n}[pI]^2 \|\boldsymbol{x}_i^*(I) - \boldsymbol{x}_i^{*,n}(I)\|_2^2}.
\]
Similarly, \(\|\boldsymbol{q}^* - \boldsymbol{q}^{r,n}\|_2 = \sqrt{\sum_{i \in \{1, 2\}} \sum_{I \in \mathcal{I}_i} \boldsymbol{q}_i^{*,n}[pI]^2 \|\boldsymbol{x}_i^*(I) - \boldsymbol{x}_i^{r,n}(I)\|_2^2}\). By the RT framework initialization (Equation~\eqref{eq:rt_init}), \(\boldsymbol{q}^{r,n} = \boldsymbol{q}^{T+1,n-1} = \boldsymbol{q}^{*,n-1}\). Thus:
\[
    \|\boldsymbol{q}^* - \boldsymbol{q}^{*,n-1}\|_2^2 - \|\boldsymbol{q}^* - \boldsymbol{q}^{*,n}\|_2^2 = \sum_{i \in \{1, 2\}} \sum_{I \in \mathcal{I}_i} \boldsymbol{q}_i^{*,n}[pI]^2 \left( \|\boldsymbol{x}_i^*(I) - \boldsymbol{x}_i^{*,n-1}(I)\|_2^2 - \|\boldsymbol{x}_i^*(I) - \boldsymbol{x}_i^{*,n}(I)\|_2^2 \right).
\]
Using the identity \(a^2 - b^2 = (a - b)(a + b)\), we obtain:
\[
    \|\boldsymbol{q}^* - \boldsymbol{q}^{*,n-1}\|_2^2 - \|\boldsymbol{q}^* - \boldsymbol{q}^{*,n}\|_2^2 \geq \sum_{i \in \{1, 2\}} \sum_{I \in \mathcal{I}_i} \boldsymbol{q}_i^{*,n}[pI]^2 \cdot \frac{C_I^2}{\mu^2 \left( \|\boldsymbol{x}_i^*(I) - \boldsymbol{x}_i^{*,n-1}(I)\|_2 + \|\boldsymbol{x}_i^*(I) - \boldsymbol{x}_i^{*,n}(I)\|_2 \right)} \geq M_{\mathcal{Q}} C_{\min},
\]
where \(M_{\mathcal{Q}} = \max_{\boldsymbol{q} \in \mathcal{Q}} \|\boldsymbol{q}\|_1\). Since \(\sum_{i \in \{1, 2\}} \sum_{I \in \mathcal{I}_i} \boldsymbol{q}_i^{*,n}[pI]^2 \leq M_{\mathcal{Q}}\), the sequence \(\{\|\boldsymbol{q}^* - \boldsymbol{q}^{*,n}\|_2\}\) is strictly decreasing and non-negative. By the monotone convergence theorem, it converges to zero, so \(\boldsymbol{q}^{*,n} \to \boldsymbol{q}^*\).

By Theorem~\ref{thm:best_iterate_convergence_of_RTCFR}, there exists \(t \in \{1, \dots, T\}\) such that \(\|\boldsymbol{q}^{*,n} - \boldsymbol{q}^{t,n}\|_2 \leq O(1/\sqrt{T})\), which vanishes as \(T \to \infty\). Thus:
\[
    \lim_{n \to \infty, T \to \infty} \|\boldsymbol{q}^* - \boldsymbol{q}^{t,n}\|_2 \leq \lim_{n \to \infty} \|\boldsymbol{q}^* - \boldsymbol{q}^{*,n}\|_2 + \lim_{T \to \infty} \|\boldsymbol{q}^{*,n} - \boldsymbol{q}^{t,n}\|_2 = 0,
\]
establishing the asymptotic last-iterate convergence of RTCFR to the \(\epsilon\)-EFPE \(\boldsymbol{q}^*\).
\end{proof}

\section{Experimental Settings}
\subsection{Game Descriptions}\label{a:game_description}

\textbf{Kuhn Poker} \cite{kuhn1950simplified} is a simplified poker variant with a deck of \(n\) cards (e.g., King, Queen, Jack for \(n=3\)). Each player receives one card, with the rest hidden. Players can check, raise, call, or fold, and the highest card wins the pot. Larger \(n\) increases strategic complexity.

\textbf{Leduc Poker} \cite{southey2012bayes} uses \(2n\) cards (e.g., two Kings, Queens, Jacks for \(n=3\)). Each player gets a private card, and a public card is revealed after the first betting round. A second betting round follows, and the player with a card matching the public card’s rank wins; otherwise, the highest card prevails.

\textbf{Goofspiel} \cite{ross1971goofspiel} involves three decks of \(n\) cards (values 1 to \(n\)). Each player receives one deck, and the third (“prize”) deck is shuffled face-down. Each turn, the top prize card is revealed, players bid with a card from their hand, and the highest bid wins the prize. Ties split the prize. The score is the difference in prize values won.

\textbf{Liar’s Dice} \cite{lisy2015online} starts with each player rolling a \(n\)-sided dice privately. Players alternate claiming a number (1 to \(n\)) and a minimum quantity of dice showing that number across all rolls. A player can raise the claim or challenge the prior claim. A false claim awards +1 to the challenger and -1 to the claimant; a true claim reverses this.

\subsection{Detailed Experimental Settings}\label{a:exp_setting}

For Reg-OMWU (0), Reg-OMWU (0.001), and Reg-OMWU (adp), step sizes (\(\eta\)) were tuned via logarithmic grid search from \(10^{-5}\) to 1 over 20 points, as shown in Table \ref{tab:omwu_par_in_EFG}. OMWU and EGT algorithms employ the dilated technique \cite{hoda2010smoothing} for bottom-up optimization of behavior strategies, with regularizer weights \(\alpha_I = 2 + 2 \max_{a \in A(I)} \sum_{j \in C(I,a)} \alpha_j\) for information set \(I \in \mathcal{I}_i\), where \(C(I,a)\) denotes child information sets reached via action \(a\) \cite{kroer2017theoretical}.

\begin{table}[ht]
    \centering
    \caption{Hyperparameters (\(\eta\)) for OMWU variants in EFGs}
    \begin{tabular}{lccccccc}
        \toprule
        Game & Kuhn (3) & Leduc (3) & Leduc (5) & Goofspiel (3) & Goofspiel (4) & Liar’s Dice (5) & Liar’s Dice (6) \\
        \midrule
        \(\eta\) & 0.0885 & 0.0263 & 0.0042 & 0.0023 & 0.0012 & 0.0001 & \(3.36 \times 10^{-5}\) \\
        \bottomrule
    \end{tabular}
    \label{tab:omwu_par_in_EFG}
\end{table}

Reg-CFR \cite{liu2022power}, using Dual Stabilized Optimistic Mirror Descent \cite{hsieh2021adaptive}, was tested but performed poorly due to its \(O(\sqrt{t})\) step size growth, inefficient for large game trees. We adopted Reg-OMWU \cite{liu2022power} instead, which uses OMWU for local optimization.

For Reg-OMWU (adp), we modified the adaptive regularization and perturbation from \cite{bernasconi2024learning}. The original method required \(10^7\) iterations and used an inefficient initial perturbation (\(1 - 10^{-4}\)). We set the initial regularization weight to \(10^{-5}\), halving it when exploitability decreases by 0.245, and tuned perturbations from 0.1 to 0.001 via a 100-point logarithmic grid, ensuring the regularization weight remains smaller than the perturbation for effective EFPE computation.

\end{document}